\documentclass[smallextended]{svjour3}

\smartqed % flush right qed marks, e.g. at end of proof

\usepackage{graphicx}
\usepackage{mathptmx}      % use Times fonts if available on your TeX system

\usepackage{graphicx}
\usepackage{amssymb,amsmath,stmaryrd, mathrsfs}     %%% MATHEMATICAL SYMBOLS
\usepackage{subfigure}
\usepackage{latexsym}
\usepackage[square]{natbib}

\newcommand\tuple[1]{\langle #1 \rangle}
\newcommand{\EDF}{\textsc{edf}}
\newcommand{\GEDF}{g\textsc{edf}}
\newcommand{\GLLF}{g\textsc{llf}}

\newcommand{\EDZL}{\textsc{edzl}}
\newcommand{\USEDF}{\textsc{us-edf}}
\newcommand{\RMUS}{\textsc{rm-us}}
\newcommand{\FPEDF}{\textsc{fp-edf}}
\newcommand{\RM}{\textsc{rm}}
\newcommand{\DM}{\textsc{dm}}

\DeclareMathOperator{\dbf}{dbf}
\DeclareMathOperator{\sbf}{sbf}
\DeclareMathOperator{\lsbf}{lsbf}
\DeclareMathOperator{\usbf}{usbf}

\DeclareMathOperator{\GCD}{GCD}

\DeclareMathOperator{\dem}{DEM}

\DeclareMathOperator{\component}{\mathcal{C}}

\DeclareMathOperator{\taskset}{\mathcal{T}}
\DeclareMathOperator{\task}{\tau}
\DeclareMathOperator{\Tperiod}{T}
\DeclareMathOperator{\Tcapacity}{C}
\DeclareMathOperator{\Tdeadline}{D}
\DeclareMathOperator{\Jcapacity}{c}

\DeclareMathOperator{\Rperiod}{\Pi}
\DeclareMathOperator{\Rcapacity}{\Theta}
\DeclareMathOperator{\mpr}{\mu}

\newcommand{\Rprocessors}{m'}

\newcommand\GSTask{\task = (\Tperiod,\Tcapacity,\Tdeadline)}  
\newcommand\STask[1]{\task_{#1} =
  (\Tperiod_{#1},\Tcapacity_{#1},\Tdeadline_{#1})}
\newcommand\GMpr{\mpr =
  \tuple{\Rperiod,\Rcapacity,\Rprocessors}}
\newcommand\Mpr[1]{\mpr_{#1} =
  \tuple{\Rperiod_{#1},\Rcapacity_{#1},\Rprocessors_{#1}}} 

\journalname{Real-Time Systems}

\begin{document}

%----------------------------------------------------------------------------
%body of document

\title{Optimal Virtual Cluster-based Multiprocessor
  Scheduling\thanks{This research was supported in part by AFOSR
    FA9550-07-1-0216, NSF CNS-0509327, NSF CNS-0720703, ONR MURI
    N00014-07-1-0907, NSF CNS-0721541 and NSF CNS-0720518. This
    research was also supported in part by IT R\&D program of MKE/KEIT
    of Korea [2009-F-039-01], KAIST Institute of Design of Complex
    Systems and KAIST-Microsoft Research Collaboration
    Center.}\thanks{This is an extended version of the ECRTS'08
    paper~\citep{SEL08}.}}
\author{
Arvind Easwaran
\and 
Insik Shin
\and
Insup Lee
}

\institute{Arvind Easwaran \at
Department of CIS, University of Pennsylvania, PA, 19104, USA. \\
\emph{Current affiliation: } CISTER/IPP-HURRAY, Polytechnic Institute
of Porto, Portugal. \\ 
Tel.: +351-22-834-0529, Fax.: +351-22-834-0509,
\email{aen@isep.ipp.pt}
\and
Insik Shin \at
Department of Computer Science, KAIST, Daejeon, Republic of Korea. \\
Tel.: +82-42-350-3524, Fax.: +82-42-350-3510, \email{insik.shin@cs.kaist.ac.kr}
\and
Insup Lee \at
Department of CIS, University of Pennsylvania, PA, 19104, USA. \\
Tel.: +1-215-898-3532, Fax.: +1-215-573-7362, \email{lee@cis.upenn.edu}
}

\date{Received: date / Accepted: date}
% The correct dates will be entered by the editor

\maketitle

  \begin{abstract}
Scheduling of constrained deadline sporadic task systems on 
multiprocessor platforms is an 
area which has received much attention in the recent past. It is
widely believed that finding an optimal scheduler is hard, and
therefore most studies have focused on developing algorithms with
good processor utilization bounds. These algorithms can be broadly
classified into two categories: partitioned scheduling in which tasks
are statically assigned to individual processors, and global
scheduling in which each task is allowed to execute on any processor
in the platform. In this paper we consider a third, more general,
approach called cluster-based scheduling. In this approach each task is
statically assigned to a processor cluster, tasks in each cluster are
globally scheduled among themselves, and clusters in turn are
scheduled on the multiprocessor platform. We develop techniques to
support such cluster-based scheduling algorithms, and also consider
properties that minimize total processor utilization of individual
clusters. In the last part of this paper, we develop new virtual
cluster-based scheduling algorithms. For implicit deadline sporadic
task systems, we develop an optimal scheduling algorithm that
is neither Pfair nor ERfair. We also show that the processor 
utilization bound of \USEDF$\{m/(2m-1)\}$ can be improved by using
virtual clustering. Since neither partitioned nor global strategies
dominate over the other, cluster-based scheduling is a natural
direction for research towards achieving improved processor utilization
bounds.
% \PACS{PACS code1 \and PACS code2 \and more}
% \subclass{MSC code1 \and MSC code2 \and more}
% \CRclass{}
\end{abstract}

\section{Introduction}
\label{sec:mpr:introduction}

With rapid development in microprocessor technology, multiprocessor and
multi-core designs are becoming an attractive solution to fulfill
increasing performance demands. In the real-time systems community,
there has been a growing interest in multiprocessor scheduling
theories. In general, existing
approaches over $m$ identical, unit-capacity processors can fall into
two categories: {\em partitioned} and {\em global} scheduling. Under
partitioned scheduling each task is statically assigned to a single
processor and is allowed to execute on that processor only. Under
global scheduling tasks are allowed to dynamically migrate across $m$
processors and execute on any of them. 

In this paper we consider another approach using a notion of {\em
 processor cluster}. A cluster is a set of $\Rprocessors$
processors, where $1 \leq \Rprocessors \leq m$. Under cluster-based
scheduling, tasks are statically assigned to a cluster and then 
globally scheduled within the cluster. This scheduling strategy can be
viewed as a generalization of partitioned and global scheduling; it is
 equivalent to partitioned scheduling at one extreme end where we
assign tasks to $m$ clusters each of size one, and global scheduling
at the other extreme end where we assign tasks to a single cluster
of size $m$. Cluster-based scheduling can be further classified into
two types: {\em physical} and {\em virtual} depending on how a
cluster is mapped to processors in the platform. A physical cluster
holds a static one-to-one mapping between its $\Rprocessors$ processors and
some $\Rprocessors$ out of $m$ processors in the platform~\citep{CAB07}. A
virtual cluster allows a dynamic one-to-many mapping between its
$\Rprocessors$ processors and the $m$ processors in the
platform. Scheduling tasks in this virtual cluster can be viewed as 
scheduling them globally on all the $m$ processors in the platform
with amount of concurrency at most $\Rprocessors$, \emph{i.e.}, at any
time instant at most $\Rprocessors$ of the $m$ processors are used by
the cluster. A key difference is that physical clusters share no
processors in the platform, while virtual clusters can share some.

\textbf{Motivating example.} We now illustrate the capabilities of
cluster-based scheduling using an example. Consider a sporadic task
system comprised of $6$ tasks as 
follows: $\task_1 = \task_2 = \task_3 = \task_4 = (3,2,3)$, $\task_5 =
(6,4,6)$ and $\task_6 = (6,3,6)$. The notation followed here is
$(\Tperiod,\Tcapacity,\Tdeadline)$, where $\Tperiod$ denotes the minimum
release separation between successive instances of the task,
$\Tcapacity$ denotes the maximum required processor capacity for each
instance and $\Tdeadline$ denotes the relative deadline. Let
this task set be scheduled on a multiprocessor platform comprised of
$4$ processors. It is easy to see that this task set is not
schedulable under any partitioned scheduling algorithm, because no
processor can be allocated more than one
task. Figure~\ref{fig:mpr:example} shows the 
schedule of this task set under global Earliest Deadline First
(\GEDF)~\citep{Liu69}, \EDZL~\citep{CLA02}, Least Laxity First
(\GLLF)~\citep{Leung89}, \FPEDF~\citep{Baruah04} and
\USEDF$\{m/(2m-1)\}$~\citep{SrBa02} scheduling algorithms. As shown in the
figure, the task set is not schedulable under any of these
algorithms. Now consider cluster-based scheduling as follows: tasks 
$\task_1$, $\task_2$ and $\task_3$ are executed under \GLLF\ on a cluster
$\component_1$ comprised of $2$ processors, and
tasks $\task_4$, $\task_5$ and $\task_6$ are executed under \GEDF\ on
another cluster $\component_2$ comprised of $2$ processors. The 
resulting schedule is shown in Figure~\ref{fig:mpr:example}, and as
can be seen all the task deadlines are met.

In addition to being more general than physical clustering, virtual
clustering is also less sensitive to task-processor mappings. This can
be explained using the same example as above with an additional task
$\task_7 = (6,1,6)$. Just for comparison, suppose $\task_7$ is assigned
to the first cluster $\component_1$ along with tasks $\task_1$, $\task_2$ and
$\task_3$. Then physical cluster-based scheduling cannot accommodate
those two clusters on $4$ processors. On the other hand, virtual
clustering has a potential to accommodate them on 4 processors by
dynamically re-allocating slack from cluster $\component_2$ 
to cluster $\component_1$ (time interval $(5,6]$).
\begin{figure}
\centering
\includegraphics[width=0.7\linewidth]{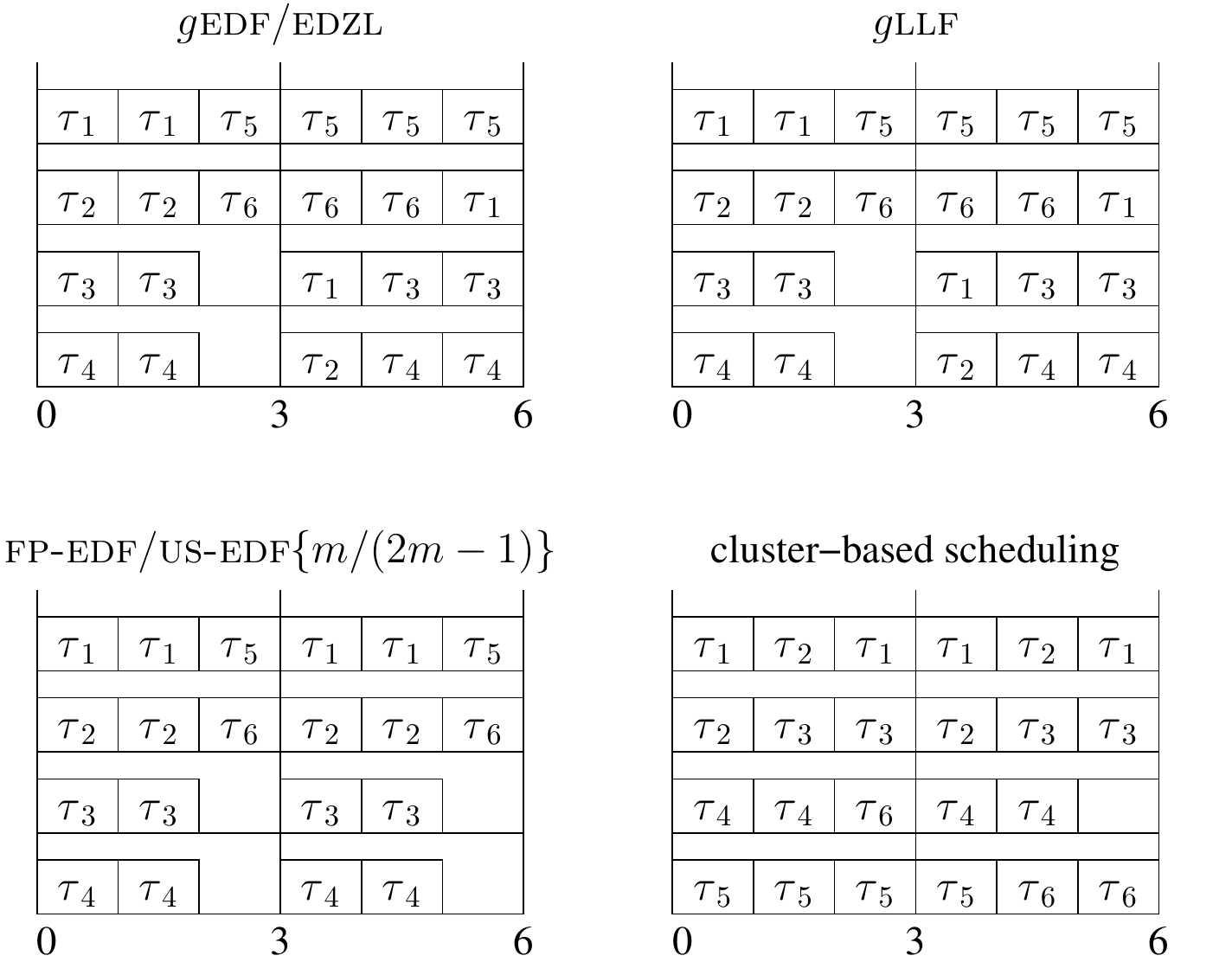}
\caption{Motivating example}
\label{fig:mpr:example}
\end{figure}

Clustering can also be useful as a mechanism to place a restriction
on the amount of concurrency. Suppose $m$ tasks can thrash a L2
cache in a multi-core platform, if they run in parallel at the same
time. Then one may consider allowing at most $\Rprocessors$ of these 
$m$ tasks to run in parallel, in order to prevent them 
from thrashing the L2 cache. This can be easily done if the $m$ tasks
are assigned to a cluster of $\Rprocessors$ processors. A similar idea
was used in~\citep{ACD06}.

\textbf{Hierarchical scheduling.} Physical clustering
requires intra-cluster scheduling only. This is because clusters are
assigned disjoint physical processors, and hence tasks in different clusters
cannot interfere with each others executions. However, the notion of
virtual clustering inherently requires a two-level hierarchical
scheduling framework; inter- and intra-cluster scheduling. In
inter-cluster scheduling physical processors are dynamically 
assigned to virtual clusters. In intra-cluster scheduling processor
allocations given to a virtual cluster are assigned to tasks 
in that cluster. Consider the example shown in
Figure~\ref{fig:mpr:hierarchy_components}. Let a task set be divided 
into three clusters $\component_1$, $\component_2$ and
$\component_3$, each employing \GEDF\ scheduling strategy. If
we use physical clustering, then each cluster can be separately
analyzed using existing techniques for \GEDF. On the other hand
if we use virtual clustering, then in addition to intra-cluster
schedulability analysis, there is a need to develop techniques for
scheduling the clusters on the multiprocessor platform. Therefore, supporting hierarchical
multiprocessor scheduling is cardinal to the successful development of
virtual clustering.
%Towards the development of an extensive
%cluster-based scheduling theory, this paper presents a preparatory
%study focusing on the hierarchical multiprocessor scheduling
%framework.
\begin{figure}
\centering
\includegraphics[width=0.6\linewidth]{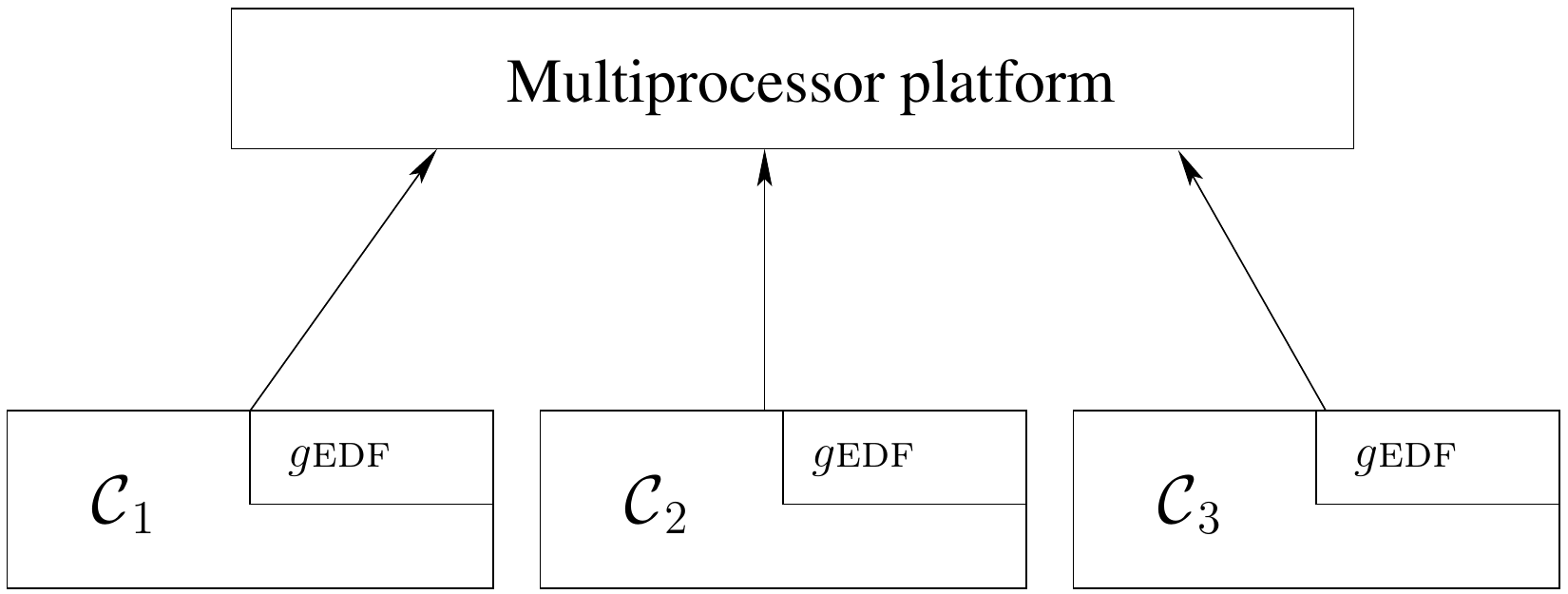}
\caption{Example virtual clustering framework}
\label{fig:mpr:hierarchy_components}
\end{figure}

There have been considerable studies on hierarchical uniprocessor
scheduling. Denoting a collection of tasks and a scheduler as a 
\emph{component}, these studies employed the notion of a component
interface to specify resources required for scheduling the
component's tasks~\citep{MFC01,ShLe03,EAL07}. Analogously, we denote a
cluster along with the tasks and scheduler assigned to it as a
\emph{component} in hierarchical multiprocessor schedulers. To support
inter-cluster scheduling, this paper proposes a component interface
that specifies resources required by the tasks in the component's
cluster. Inter-cluster scheduler can allocate processor supply to the
cluster based on its interface. Intra-cluster scheduler can then use
this processor supply to schedule the tasks in the cluster. Many new
issues arise to adopt the notion of a component interface from
uniprocessor to multiprocessor scheduling. One of them is how to
enable a component interface to carry information about concurrent
execution of tasks in the component. For example, suppose a single
task cannot execute in parallel. Then multiple processors cannot be
used concurrently to satisfy the execution requirement of this single
task. Such an issue needs to be handled for the successful development
of component interfaces. In this paper we present one solution to
this issue. Our approach is to capture in a component's interface, all
the task-level concurrency constraints in that component. The
interface demands enough processor supply from inter-cluster scheduler
so that the intra-cluster scheduler can handle task-level concurrency
constraints. As a result the inter-cluster scheduler does not have to
worry about this issue.

\textbf{Contributions.} The contributions of this paper are five-fold.
First, we introduce the notion of general hierarchical multiprocessor
schedulers to support virtual cluster-based scheduling. Second, we
present an approach to specify the task-level concurrency constraints
in a component's interface. In
Section~\ref{sec:mpr:system models} we introduce a multiprocessor
resource model based interface that not only captures the task-level
concurrency constraints, but also specifies the total resource
requirements of the component. This enables the inter-cluster
scheduler to schedule clusters using their interfaces alone. Third, 
since such interfaces represent \emph{partitioned resource
  supplies}\footnote{If a processor can be used by a cluster only in
  some time intervals and not all, then its supply is said to be
  partitioned.} as opposed to \emph{dedicated resource
  supplies}\footnote{If a processor can be used by a cluster at all
  times, then its supply is said to be dedicated.}, we also extend
existing schedulability conditions for \GEDF\ in this
direction\footnote{We have chosen to focus on one scheduling algorithm 
  in this paper. However the issues are the same for other schedulers,
  and hence techniques developed here are applicable to other
  schedulers as well.} (see
Section~\ref{sec:mpr:component_schedulability_condition}). Such
extensions to schedulability conditions are essential for supporting
development of component interfaces. Fourth, we consider the
optimization problem of minimizing the total resource requirements of
the component interface. In Section~\ref{sec:mpr:interface_generation}, we present
an efficient solution to this problem based on the following
property of our \GEDF\ schedulability condition: total processor
utilization required by a component interface to schedule tasks in the
component increases, as number of processors allocated to the
component's cluster increases. Thus an optimal solution is 
obtained when we find the smallest number of processors that 
guarantee schedulability of the component. Fifth, in
Section~\ref{sec:mpr:improved_virtual_cluster_main} we develop an  
overhead free inter-cluster scheduling framework based on
McNaughton's algorithm~\citep{McN59}. Using this framework we present 
a new algorithm, called \emph{V}irtual
\emph{C}lustering - \emph{I}mplicit \emph{D}eadline \emph{T}asks
(VC-IDT), for scheduling implicit deadline sporadic task systems on
identical, unit-capacity multiprocessor platforms. We show that VC-IDT
is an optimal scheduling algorithm, 
that does not satisfy the property of P-fairness~\citep{BCP96} or
ER-fairness~\citep{AnSr00}. The latter feature of our algorithm, as we
will see in Section~\ref{sec:VC-IDT}, translates into better bounds on the
number of preemptions. As an illustration of the capabilities of
general task-processor mappings supported by virtual clustering, we
also show that the processor utilization bound of \USEDF$\{m/(2m-1)\}$ 
can be improved by using this framework. In our previous
work~\citep{SEL08} we presented the first four contributions listed
above. In this paper we elaborate on (and extend) those contributions, and
in the process develop new virtual cluster-based 
scheduling algorithms (fifth contribution described above).

\section{Task and resource models}
\label{sec:mpr:system models}

In this section we describe our task model and the multiprocessor platform.
We also introduce multiprocessor resource models which we use as
component interfaces.

\subsection{Task and platform models}
\label{sec:mpr:task_model}

\begin{figure}
\centering
\subfigure[Case 1]{
\includegraphics[width=0.85\linewidth]{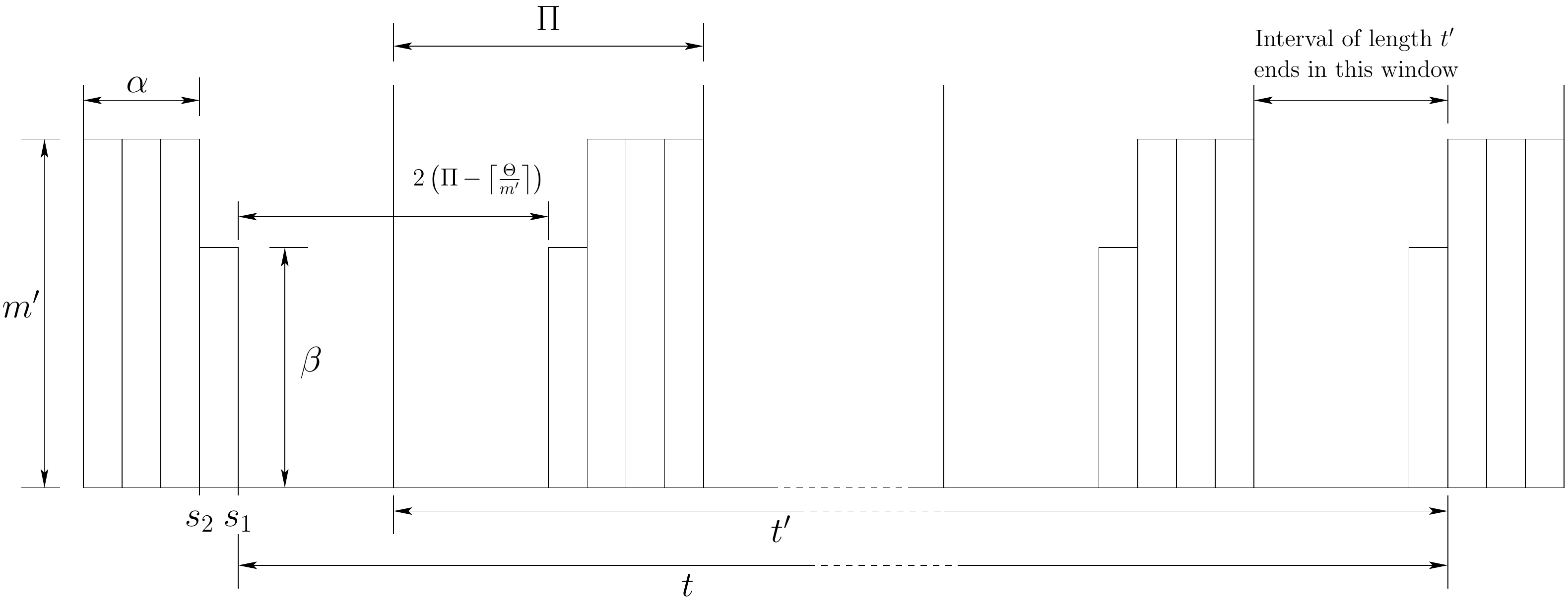}
\label{fig:mpr:sbf_schedule_MPR_0}
}
\subfigure[Case 2]{
\includegraphics[width=0.85\linewidth]{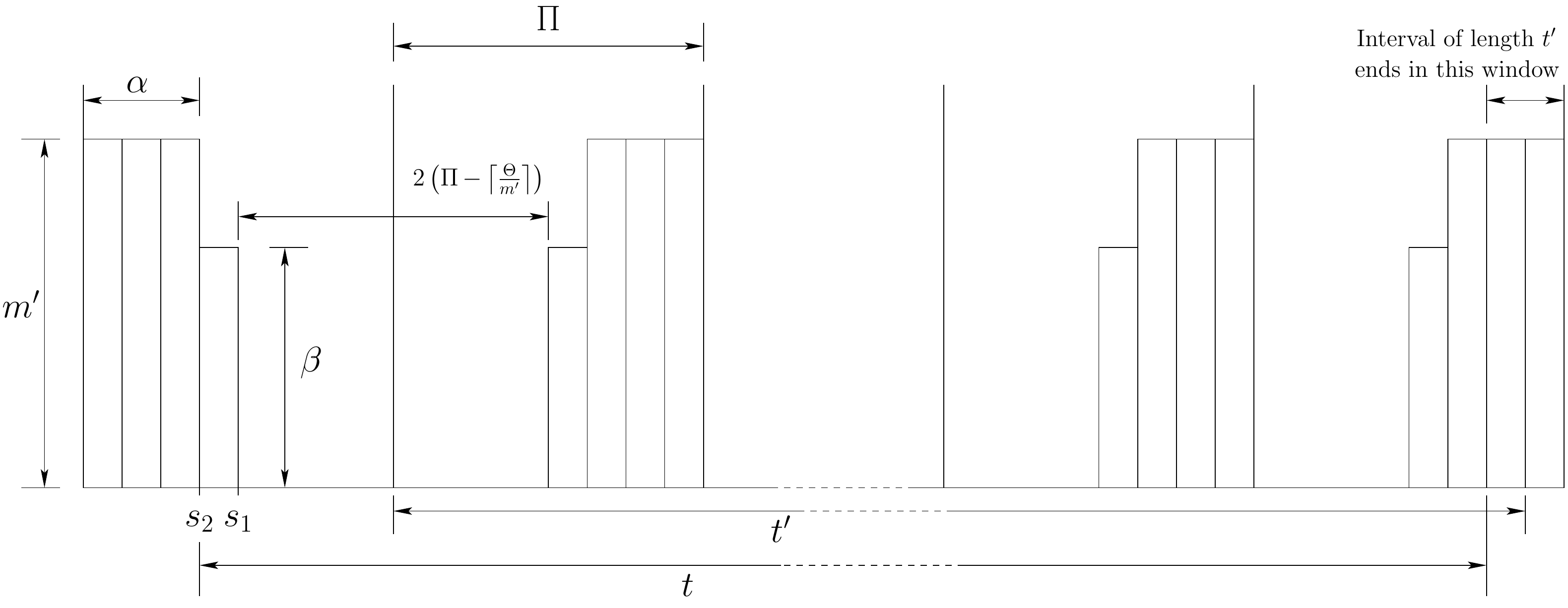}
\label{fig:mpr:sbf_schedule_MPR_1}
}
\caption{Schedule of $\mpr$ w.r.t $\sbf_{\mpr}(t)$}
\label{fig:mpr:sbf_schedule_MPR}
\end{figure}
\textbf{Task model.} We assume a constrained deadline sporadic task
model~\citep{BMR90}. In this model a sporadic task is specified as
$\STask{i}$, where $\Tperiod_i$ is the minimum release separation,
$\Tcapacity_i$ is the maximum processor capacity requirement and
$\Tdeadline_i$ is the relative deadline. These task parameters satisfy
the property $\Tcapacity_i \leq \Tdeadline_i \leq
\Tperiod_i$\footnote{If $\Tdeadline_i = \Tperiod_i$ then the task is 
  called implicit deadline task.}. Successive instances of
$\task_i$ are released with a minimum separation of $\Tperiod_i$ time
units. We refer to each 
such instance as a \emph{real-time job}. Each job of $\task_i$ must
receive $\Tcapacity_i$ units of processor capacity within
$\Tdeadline_i$ time units from its release. These $\Tcapacity_i$ units
must be supplied sequentially to the job. This restriction is
useful in modeling many real-world systems, because in general, all
portions of a software program cannot be parallelized. 

\textbf{Multiprocessor platform and scheduling strategy.} In this
paper we assume an identical, unit-capacity multiprocessor
platform having $m$ processors. Each processor in this platform has a
resource bandwidth of one, \emph{i.e.}, it can provide $t$ units of
processor capacity in every time interval of length $t$. We also
assume that a job 
can be preempted on one processor and may resume execution on another
processor with negligible preemption and migration overheads, as in
the standard literature of global
scheduling~\citep{GFB03,Bak05,BCL05,Baruah07}. We assume such a global
scheduling strategy within each cluster, and in particular, we assume
that the strategy is global \EDF\ (denoted as \GEDF). At
each time instant, if $\Rprocessors$ denotes the number of physical
processors allocated to the cluster, then \GEDF\ schedules unfinished
jobs that have the $\Rprocessors$ earliest relative deadlines.

\subsection{Multiprocessor resource model}
\label{sec:mpr:multi-core_resource_models}

A \emph{resource model} is a model for specifying the characteristics
of processor supply. When these models represent component interfaces,
they specify total processor requirements of the
component. Periodic~\citep{ShLe03}, EDP~\citep{EAL07},
bounded-delay~\citep{FeMo02}, etc., are examples of resource models that 
have been extensively used for analysis of hierarchical uniprocessor 
schedulers. These resource models can also be used as component interfaces 
in hierarchical multiprocessor schedulers. One way to achieve this is
to consider $\Rprocessors$ identical resource models as a component
interface, where $\Rprocessors$ is the number of processors allocated to
the component's cluster. However, this interface is restrictive because
each processor contributes the same amount of resource to the component as any other
processor in the cluster. It is desirable to be more flexible in that
interfaces should be able to represent the collective processor 
requirements of clusters, without fixing the contribution of each processor a
priori. Apart from increased flexibility, such interfaces can also
improve processor utilization in the system. 

We now introduce a \emph{multiprocessor resource} model that
specifies the characteristics of processor supply provided 
by an identical, unit-capacity multiprocessor platform. This resource
model does not fix the contribution of each processor a
priori, and hence is a suitable candidate for cluster interfaces. 
% \begin{definition}[Multiprocessor periodic resource model (MPR)]
% A \emph{multiprocessor periodic resource} (MPR) model $\GMpr$ specifies
% that an identical, unit-capacity multiprocessor platform collectively  
% provides $\Rcapacity$ units of resource in every $\Rperiod$ time units
% to a cluster comprising of $\Rprocessors$ processors. These 
% $\Rcapacity$ resource units are supplied with concurrency at most
% $\Rprocessors$, \emph{i.e.}, at any time instant at most
% $\Rprocessors$ physical processors are allocated to this resource
% model. It is then easy to see that a \emph{feasible MPR model} must
% satisfy the condition $\Rcapacity \leq \Rprocessors \Rperiod$. Also,
% $\frac{\Rcapacity}{\Rperiod}$ denotes the resource bandwidth of model $\mpr$.
% \end{definition}
\begin{definition}[Multiprocessor periodic resource model (MPR)]
A \emph{multiprocessor periodic resource} model $\GMpr$ specifies
that an identical, unit-capacity multiprocessor platform collectively  
provides $\Rcapacity$ units of resource in every $\Rperiod$ time
units, where the $\Rcapacity$ units are supplied with concurrency at most
$\Rprocessors$; at any time instant at most $\Rprocessors$ physical
processors are allocated to this resource
model. $\frac{\Rcapacity}{\Rperiod}$ denotes the \emph{resource
  bandwidth} of model $\mpr$.
\end{definition}

It is easy to see from the above definition that a \emph{feasible MPR
  model} must satisfy the condition $\Rcapacity \leq \Rprocessors
\Rperiod$. The supply bound function of a resource model ($\sbf$)
lower bounds the amount of processor supply that the model guarantees
in a given time interval. Specifically, $\sbf_R(t)$ is equal to the
minimum amount of processor capacity that model $R$ is guaranteed to provide in
any time interval of duration $t$. In uniprocessor systems, $\sbf$ is
used in schedulability conditions to generate resource model based
component interfaces. Extending this approach to multiprocessors, in
this paper we derive similar schedulability conditions to generate MPR 
model based component interfaces. Hence we now present the $\sbf$ for
a MPR model $\GMpr$. Figure~\ref{fig:mpr:sbf_schedule_MPR} shows the
schedule for $\mpr$ that generates this minimum supply in a time
interval of duration $t$, where $\alpha = \left \lfloor
\frac{\Rcapacity}{\Rprocessors} \right \rfloor$ and $\beta = \Rcapacity -
\Rprocessors \alpha$. As can be seen, length of the largest time
interval with no supply is equal to $2\Rperiod - 2 \left \lceil
  \frac{\Rcapacity}{\Rprocessors} \right \rceil$ (shown in
the figures). $\sbf_{\mpr}$\footnote{A correction has been made to
  $\sbf_{\mpr}$ from its original publication in~\citep{SEL08}.} is
given by the following equation.
\begin{equation*}
\sbf_{\mpr}(t) = 
\begin{cases}
 0 & t' < 0 \\
 \left \lfloor \frac{t'}{\Rperiod} \right \rfloor  \Rcapacity + \max
 \left \{0, \Rprocessors x - \left (\Rprocessors \Rperiod - \Rcapacity
   \right ) \right \} & t' \geq 0 \mbox{ and } x \in \left [ 1, y \right ] \\ 
 \left \lfloor \frac{t'}{\Rperiod} \right \rfloor  \Rcapacity + \max
 \left \{0, \Rprocessors x - \left (\Rprocessors \Rperiod - \Rcapacity
   \right ) \right \} - (\Rprocessors - \beta) & t' \geq 0 \mbox{ and
 } x \not \in \left [ 1, y \right ]
\end{cases}
\end{equation*}
\begin{equation}
\mbox{ where } t' = t - \left ( \Rperiod - \left \lceil
    \frac{\Rcapacity}{\Rprocessors} \right \rceil \right ) \mbox{, } x = \left (
  t' - \Rperiod \left \lfloor \frac{t'}{\Rperiod} \right \rfloor
\right ) \mbox{ and } y = \Rperiod - \left \lfloor
  \frac{\Rcapacity}{\Rprocessors} \right \rfloor
\label{eqn:mpr:sbf_MPR}
\end{equation}
% \begin{equation*}
% \sbf_{\mpr}(t) = 
% \begin{cases}
%  0 & t' < 0 \\
%  \left \lfloor \frac{t'}{\Rperiod} \right \rfloor \Rcapacity + \Gamma
% & t' \geq 0 \mbox{ and } t' - \Rperiod \left \lfloor
%        \frac{t'}{\Rperiod} \right \rfloor \in \left [ 1, \Rperiod -
%        \left \lfloor \frac{\Rcapacity}{\Rprocessors} \right \rfloor \right ] \\
%  \left \lfloor \frac{t'}{\Rperiod} \right \rfloor \Rcapacity + \Gamma
%  - (\Rprocessors - \beta) & t' \geq 0 \mbox{ and } t' - \Rperiod \left
%    \lfloor \frac{t'}{\Rperiod} \right \rfloor \not \in \left [ 1,
%    \Rperiod - \left \lfloor \frac{\Rcapacity}{\Rprocessors} \right
%    \rfloor \right ]
% \end{cases}
% \end{equation*}
% \begin{equation}
% \mbox{ where } t' = t - \left ( \Rperiod -
%   \left \lceil \frac{\Rcapacity}{\Rprocessors} \right \rceil \right
% ) \mbox{ and } \Gamma = \max \left \{ 0, \Rprocessors \left ( t' -
%     \Rperiod \left \lfloor \frac{t'}{\Rperiod} \right \rfloor \right )
%   - (\Rprocessors \Rperiod - \Rcapacity) \right \} 
% \label{eqn:mpr:sbf_MPR}
% \end{equation}
There are two main cases to consider for $\sbf_{\mpr}$. If $t'$ is as
shown in Figure~\ref{fig:mpr:sbf_schedule_MPR_0}, then the interval that
generates the minimum supply starts from time instant $s_1$ shown in
the same figure. On the other hand, if $t'$ is as shown in
Figure~\ref{fig:mpr:sbf_schedule_MPR_1}, then the interval that
generates the minimum supply starts from time instant $s_2$ shown in
the same figure.  
In uniprocessor systems although schedulability conditions with 
$\sbf$ have been derived, a linear approximation of $\sbf$ is often
used to reduce the time-complexity of the interface generation
process. Hence, in anticipation, we present the following linear lower
bound for $\sbf_{\mpr}$\footnote{$\lsbf_{\mpr}$ has also been modified
  from its original publication in~\citep{SEL08}, in order to be consistent with
  the new $\sbf_{\mpr}$.}. Functions $\sbf_{\mpr}$ and $\lsbf_{\mpr}$
are plotted in Figure~\ref{fig:mpr:sbf_MPR}. 
\begin{align}
 \lsbf_{\mpr}(t) = & \frac{\Rcapacity}{\Rperiod} \left (t - \left [ 2 \left
     (\Rperiod - \frac{\Rcapacity}{\Rprocessors} \right ) + 2
 \right ] \right ) 
\label{eqn:mpr:lsbf_MPR_2}
\end{align}

The following lemma proves that $\lsbf_{\mpr}$ is indeed a lower bound
for $\sbf_{\mpr}$.
\begin{lemma}
$\lsbf_{\mpr}(t) \leq \sbf_{\mpr}(t)$ for all $t \geq 0$.
\end{lemma}
\begin{proof}
Consider Figure~\ref{fig:mpr:sbf_MPR}. Observe that $\lsbf_{\mpr}(t) =
0$ for all $t \leq t_4$. Therefore it is sufficient to show that
$\lsbf_{\mpr}(t) \leq \sbf_{\mpr}(t)$ for all $t > t_4$. Suppose
$\sbf_{\mpr}(t_4) = 2\beta + \epsilon$ for some $\epsilon \geq
0$. 

We now show that $\lsbf_{\mpr}(t) \leq \sbf_{\mpr}(t)$ for all $t$
such that $t_4 < t \leq t_8$, where $t_8 = t_4 + \Rperiod$. The
following statements are true by definition: 1) $\sbf_{\mpr}(t_8) =
\Rcapacity + 2\beta + \epsilon$, and 2) $\lsbf_{\mpr}(t_8) =
\Rcapacity$. Further, because the slope of $\sbf_{\mpr}$ in the
interval $(t_4, t_5]$ is at least as much as the slope of
$\lsbf_{\mpr}$ ($\frac{\Rcapacity}{\Rperiod} \leq \Rprocessors$),
$\lsbf_{\mpr}(t) \leq \sbf_{\mpr}(t)$ for all $t$ 
such that $t_4 < t \leq t_5$. From the figure, we can see that
$\sbf_{\mpr}(t_6)=\Rcapacity=\lsbf_{\mpr}(t_8)$ and $t_6 \leq
t_8$. Therefore $\lsbf_{\mpr}(t) \leq \sbf_{\mpr}(t)$ for all $t$ such 
that $t_6 \leq t \leq t_8$. The last statement follows from the fact
that $\sbf_{\mpr}$ is a non-decreasing function. This combined with
the facts that $t_6=t_5+1$ and $\lsbf_{\mpr}$ is a linear function,
implies $\lsbf_{\mpr}(t) \leq \sbf_{\mpr}(t)$ for all $t$ such that
$t_4 < t \leq t_8$.

Observe that in every successive time interval of length $\Rperiod$
starting from $t_4$, the following holds: 1) both $\sbf_{\mpr}$ and
$\lsbf_{\mpr}$ increase by exactly $\Rcapacity$, and 2) they both have
slope characteristics identical to those in the interval $(t_4,
t_8]$. Therefore the arguments from the previous paragraph hold for each
such time interval of length $\Rperiod$. The result of the lemma then
follows. 
\qed
\end{proof}

Uniprocessor resource models, such as periodic or EDP,
allow a view that a component executes over an exclusive share of a physical
uniprocessor platform. Extending this notion, MPR models
allow a view that a component, and hence the corresponding cluster, 
executes over an exclusive share of a physical multiprocessor
platform. Although this view guarantees a minimum total
processor share given by $\sbf$, it does not enforce any distribution
of this share over the processors in the platform, apart from the 
concurrency bound $\Rprocessors$. In this regard MPR models are general and
hence our candidate for component interfaces.
\begin{figure}
\centering
\includegraphics[width=\linewidth]{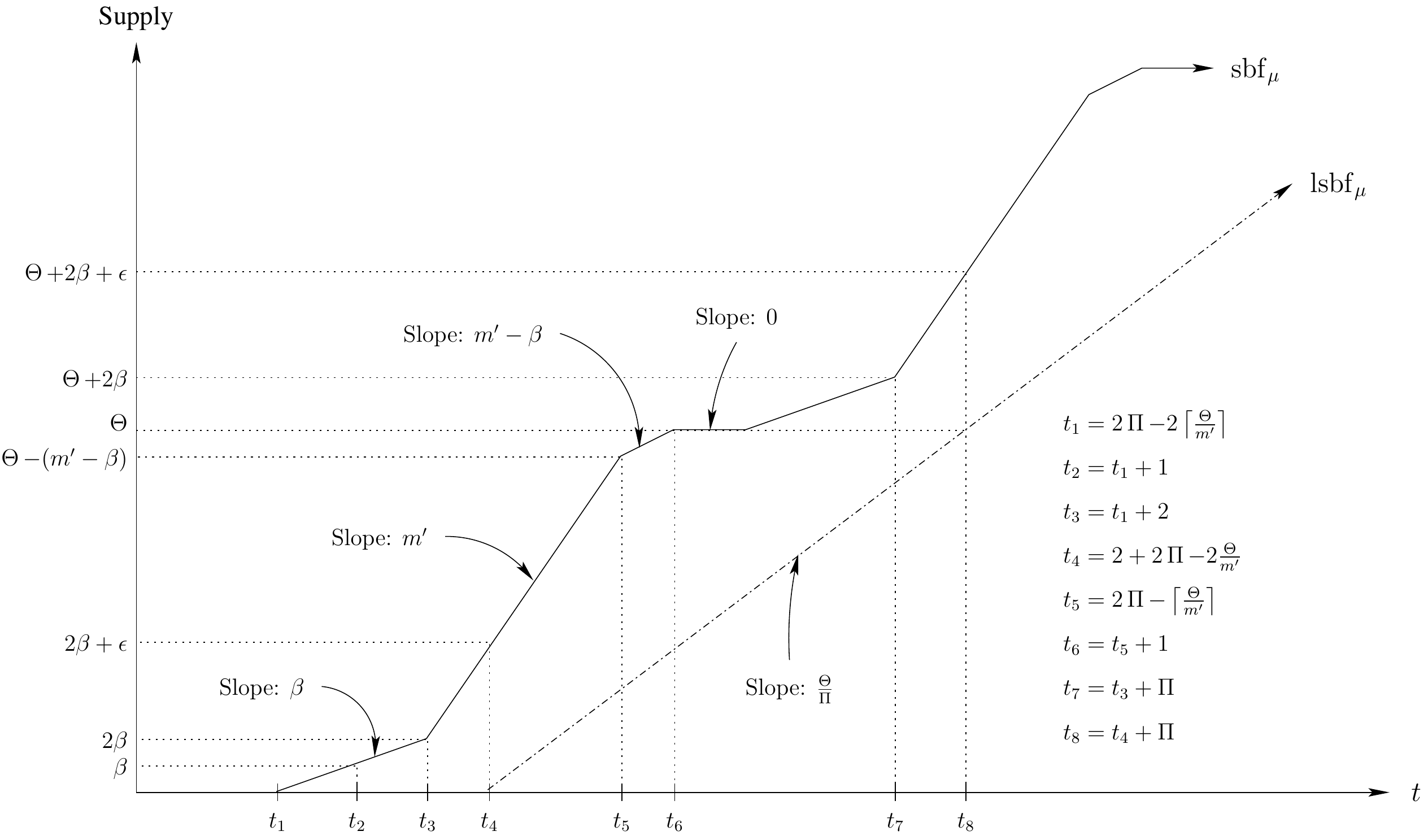}
\caption{$\sbf_{\mpr}$ and $\lsbf_{\mpr}$}
\label{fig:mpr:sbf_MPR}
\end{figure}

% For uniprocessor frameworks, EDP resource models generalize periodic
% resource models~\citep{EAL07}. Analogously, one can consider a
% generalization of MPR resource models called \emph{multiprocessor
%   explicit deadline periodic} (MEDP) resource models. A MEDP resource
% model $\Omega = \tuple{\Rperiod,\Rcapacity,\Delta,m}$ specifies that an
% identical unit-capacity multiprocessor platform collectively provides
% $\Rcapacity$ units of resource in $\Delta$ units of time with this supply
% pattern repeating every $\Rperiod$ time units. Again, the amount of
% concurrency in this processor supply is bounded by $m$.
% %In this paper, we focus on MEDP models whose deadlines are at most their
% %periods, i.e., $\Delta \leq \Rperiod$. For such models, it is easy to see
% %that $m \Delta$ is an upper bound on their capacity $\Rcapacity$.
% %We also assume that period $\Rperiod$ and deadline $\Delta$ are positive integers.
% Although MEDP generalize MPR resource models, we use MPR resource
% models as component interfaces in this paper. This choice is motivated
% only from the perspective of simplicity of presentation. Presence of
% an additional deadline parameter in MEDP complicate interface
% generation procedures. Since our goal is to address the problem of
% component interface generation towards supporting cluster-based
% scheduling and capturing task-level concurrency constraints, we use
% the more restrictive MPR resource model.

\section{Related work}
\label{sec:mpr:related}

\textbf{Multiprocessor scheduling.} In general, studies on real-time
multiprocessor scheduling theory can fall into two categories: {\em
  partitioned} and {\em global} scheduling. Under partitioned
scheduling each task is statically assigned to a single processor 
and uniprocessor scheduling algorithms are used to schedule
tasks. Under global scheduling tasks are allowed to migrate across
processors and algorithms that simultaneously schedule on all the
processors are used. Many partitioning algorithms and their
analysis~\citep{OhBa98,LDG01,BaFi06b,FBB06}, and global scheduling
algorithms and their
analysis~\citep{BCP96,ABJ01,CLA02,SrBa02,ZMM03,GFB03,Bak03,Baruah04,Bak05,Bak06,BCL05,CRJ06,Baruah07,CiBa07,BeCi07,BaFi07,BaBa08,BaBa08a,FKY08}, 
have been proposed in the past.

For implicit deadline task systems, both Earliest Deadline First (\EDF)~\citep{LDG01} and
Rate Monotonic (\RM)~\citep{OhBa98} based partitioned scheduling have been proposed
along with processor utilization bounds. These studies have since been
extended for constrained deadline task systems, and
\EDF~\citep{BaFi06b} and fixed-priority~\citep{FBB06} based scheduling
have been developed for them. Under global scheduling of implicit
deadline task systems, several optimal algorithms such as
Pfair~\citep{BCP96}, BoundaryFair~\citep{ZMM03}, LNREF~\citep{CRJ06}, and
NVNLF~\citep{FKY08}, have been proposed. To reduce the relatively high
preemptions in these algorithms and to support constrained deadline task
systems, processor utilization bounds and worst-case response time
analysis for
\EDF~\citep{GFB03,Bak03,Bak05,BCL05,Baruah07,BeCi07,BaBa08,BaBa08a} and
Deadline Monotonic (\DM)~\citep{Bak03,Bak06,BeCi07,BaFi07} based global scheduling strategies have
been developed. Towards better processor utilization, new global
algorithms such as dynamic-priority
\EDZL~\citep{CLA02,CiBa07} and \USEDF$\{m/(2m-1)\}$~\citep{SrBa02}, and
fixed-priority \RMUS$\{m/(3m-2)\}$~\citep{ABJ01} and
\FPEDF~\citep{Baruah04}, have also been proposed. Partitioned
scheduling suffers from an inherent performance limitation in that a
task may fail to be assigned to any processor, although the total
available processor capacity across the platform is larger than the
task's requirements. Global scheduling has been developed to overcome
this limitation. However global algorithms are either not known to
utilize processors optimally (like in the case of constrained deadline
task systems), or if they are known to be optimal, then they have high
number of preemptions (like in the case of implicit deadline task
systems). Moreover, for constrained deadline tasks, simulations conducted by
Baker~\citep{Bak05a} have shown that partitioned scheduling performs
much better than global scheduling on an average. These simulations
reflect the large pessimism in current schedulability tests for global
algorithms. To eliminate the performance limitation of partitioned
scheduling and to achieve high processor utilization without incurring
high preemption costs, we consider the more general task-processor
mappings that virtual cluster-based scheduling proposes. 

Algorithms that support slightly more general task-processor mappings
than either partitioned or global scheduling have been proposed in the
past. Andersson \emph{et al.}~\citep{AnTo06,AnBl08,BBA08} and Kato and
Yamasaki~\citep{KaYa07} have developed algorithms that allow a task to
be scheduled on at most two processors in the platform. Virtual
cluster-based scheduling framework that we propose generalizes all
these task-processor mappings and therefore can lead to higher
processor utilization. Baruah and Carpenter~\citep{Baca03} introduced an
approach that restricts processor migration of jobs, in order to
alleviate the performance limitation of partitioned scheduling and the
processor migration overheads of global scheduling. It has been shown
that the worst-case processor utilization of this approach is no better
than partitioned scheduling (roughly 50\%). Calandrino {\em et
  al.}~\citep{CAB07} presented a physical clustering framework in which
tasks are first assigned to physical processor clusters and then 
scheduled globally within those clusters. They experimentally
evaluated this framework to show that cache-access related overheads
can be reduced in comparison to both partitioned and global scheduling
strategies. Virtual clustering is again a generalization of this
framework, and moreover, unlike their work, we develop efficient
schedulability analysis techniques with a focus on achieving high
processor utilization. Recently, virtual clustering has also been
considered in the context of tardiness guarantees for soft real-time
systems~\citep{LeAn08}.

Moir and Ramamurthy~\citep{MoRa99} and Anderson \emph{et
  al.}~\citep{HoAn01,ACD06} presented an approach that upper bounds
the amount of concurrent execution within a group of tasks. They
developed their approach using a two-level Pfair-based scheduling
hierarchy. These studies are most related to our work on virtual
clustering, but they differ from our technique mainly in the following
aspect. We introduce a multiprocessor resource model that makes it
possible to clearly separate intra- and inter-cluster scheduling. This
allows development of schedulability analysis techniques for virtual
clustering that are easily extensible to many different
schedulers. However their approaches do not employ such a
notion. Therefore their analysis techniques are bound to Pfair 
scheduling, and do not generalize to other algorithms and task
models such as the one considered in this paper. This flexibility
provides a powerful tool for the development of various task-processor
mappings and intra- and inter-cluster scheduling algorithms.

\textbf{Hierarchical scheduling.} For uniprocessor platforms there
has been a growing attention to 
hierarchical scheduling frameworks. Since a two-level framework was
introduced~\citep{DeLi97}, its schedulability has been analyzed under 
fixed-priority~\citep{KuLi99} and \EDF-based~\citep{LCB00}
scheduling. For multi-level frameworks many resource model based
component interfaces such as bounded-delay~\citep{MFC01,ShLe04},
periodic~\citep{LiBi03,ShLe03,ShLe08} and EDP~\citep{EAL07}, have 
been introduced, and schedulability conditions have been derived under
fixed-priority and \EDF\ 
scheduling~\citep{FeMo02,LiBi03,ShLe03,AlPe04,DaBu05,EAL07}. As
discussed in the introduction, these studies do not provide any
technique to capture task-level concurrency constraints in interfaces,
and therefore are not well suited for virtual clustering.

\section{Component schedulability condition}
\label{sec:mpr:component_schedulability_condition}

In this section we develop a schedulability condition for
components in hierarchical multiprocessor schedulers, such that this
condition accommodates the notion of a partitioned resource supply.
Specifically, we extend existing \GEDF\ schedulability
conditions for dedicated resource, with the supply bound function of
a MPR model. Any MPR model that satisfies this condition can be
used as an interface for the component.

We consider a component comprising of cluster $\component$ and 
sporadic tasks $\taskset = \{ \STask{1}, \ldots , \STask{n} \}$
scheduled under \GEDF. To keep the presentation simple, we use
notation $\component$ to refer to the component as well. We now
develop a schedulability condition for $\component$ assuming it is 
scheduled using MPR model $\GMpr$, where $\Rprocessors$ denotes number
of processors in the cluster. This condition uses the total processor
demand of task set $\taskset$ for a given time interval. Existing 
studies~\citep{BCL05} have developed an upper bound for this demand
which we can use. Only upper bounds are known for this demand, because
unlike the synchronous arrival sequence in uniprocessors, no notion of
worst-case arrival sequence is known for
multiprocessors~\citep{Baruah07}. Hence we first summarize this  
existing demand upper bound and then present our schedulability condition. 

\subsection{Component demand}
\label{sec:mpr:higher_priority_interference}

\textbf{Workload.} The workload of a task $\task_i$ in an interval
$[a,b]$ gives the cumulative length of all intervals in which 
$\task_i$ is executing, when task set $\taskset$ is scheduled under
$\component$'s scheduler. This workload consists of three parts
(illustrated in Figure~\ref{fig:mpr:higher_priority_interference}):
(1) the \emph{carry-in} demand generated by a job of $\task_i$ that is
released prior to $a$, but did not finish its execution requirements until
$a$, (2) the demand of a set of jobs of $\task_i$ that are both
released and have their deadlines within the interval, and (3) the
\emph{carry-out} demand generated by a job of $\task_i$ that is
released in the interval $[a,b)$, but does not finish its execution
requirements until $b$.
\begin{figure}
\centering
\includegraphics[width=0.85\linewidth]{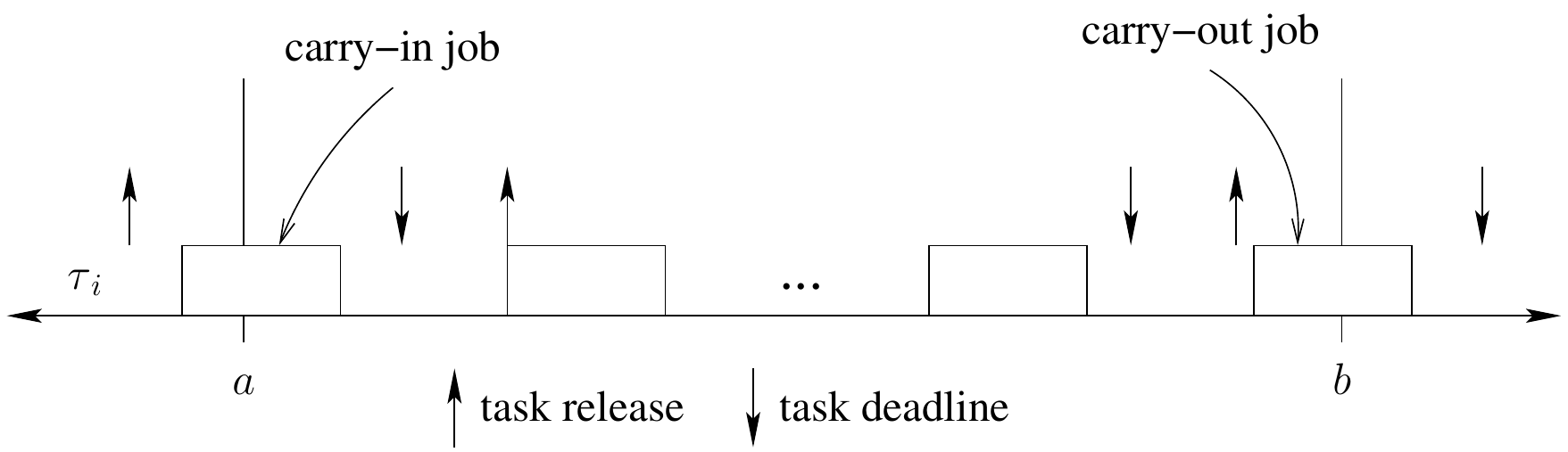}
\caption{Workload of task $\task_i$ in interval $[a,b]$}
\label{fig:mpr:higher_priority_interference}
\end{figure}

% EDF specific bound using execution pattern
\textbf{Workload upper bound for $\task_i$ under gEDF.} If workload
in an interval $[a,b]$ can be efficiently computed for all $a,b \geq
0$ and for all tasks $\task_i$, then we can obtain the exact demand
of task set $\taskset$ in all intervals. However, since no such efficient
computation technique is known (apart from task set simulation), we
use an upper bound for this workload obtained by Bertogna
\emph{et al.}~\citep{BCL05}. This bound is obtained under two
assumptions: (1) some job of some task $\task_k$ has a deadline at
time instant $b$, and (2) this job of $\task_k$ misses its
deadline. In the schedulability conditions we develop, these
assumptions hold for all time instants $b$ that are considered. Hence
this is a useful bound and we present it
here. Figure~\ref{fig:mpr:sched-proof-edf} illustrates the dispatch
pattern corresponding to this bound. A job of task $\task_i$ has a
deadline that coincides with time instant $b$. Jobs of $\task_i$ that
are released prior to time $b$ are assumed to be released as
late as possible. Also, the job of $\task_i$ that is released before
$a$ but has a deadline in the interval $[a,b]$, is
assumed to execute as late as possible. This imposes maximum
possible interference on the job of $\task_k$ with deadline
at $b$. Let $\mathcal{W}_i(t)$ denote this workload bound for $\task_i$
in a time interval of length $t$ ($= b - a$). Also let
$CI_i(t)$ denote the carry-in demand generated by the
execution pattern shown in Figure~\ref{fig:mpr:sched-proof-edf}. Then
\begin{equation*}
\mathcal{W}_i(t) = \left \lfloor \frac{t + (\Tperiod_i -
    \Tdeadline_i)}{\Tperiod_i} \right \rfloor C_i + CI_i(t),
\end{equation*}
\begin{equation}
\mbox{
  where } CI_i(t) = \min \left \{ \Tcapacity_i, \max \left \{ 0, t - \left
      \lfloor \frac{t + (\Tperiod_i - \Tdeadline_i)}{\Tperiod_i}
    \right \rfloor \Tperiod_i \right \} \right \}
\label{eqn:mpr:sched-proof-edf}
\end{equation}

 It has been shown that the actual workload of $\task_i$ can never
 exceed $\mathcal{W}_i(b-a)$ in the interval $[a,b]$, provided tasks
 are scheduled under \GEDF\ and a deadline miss occurs for that job of
 $\task_k$ whose 
 deadline is at $b$~\citep{BCL05}. This follows from the observation that
 no job of $\task_i$ with deadline greater than $b$ can execute in
 the interval $[a,b]$. In the following section we develop a
 schedulability condition for $\component$ using this workload bound.
\begin{figure}
\centering
\includegraphics[width=0.85\linewidth]{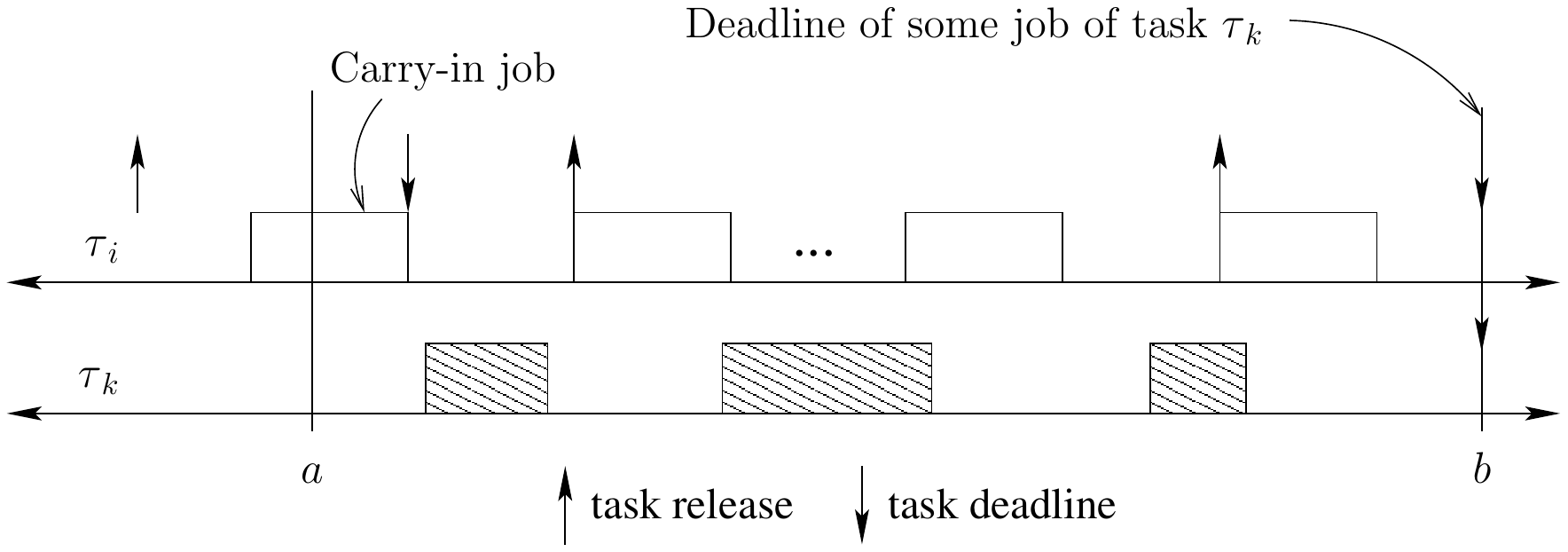}
\caption{Dispatch and execution pattern of task $\task_i$ for $\mathcal{W}_i(b-a)$}
\label{fig:mpr:sched-proof-edf}
\end{figure}

\subsection{Schedulability condition}
\label{sec:mpr:supply_interface_generation}

We now present a schedulability condition for component $\component$
when it is scheduled using MPR model $\GMpr$. For this purpose we
extend (with the notion of $\sbf_{\mpr}$) an existing condition that
checks the schedulability of $\component$ on a dedicated resource
comprised of $\Rprocessors$ unit-capacity processors.

\begin{figure}
\centering
\includegraphics[width=0.9\linewidth]{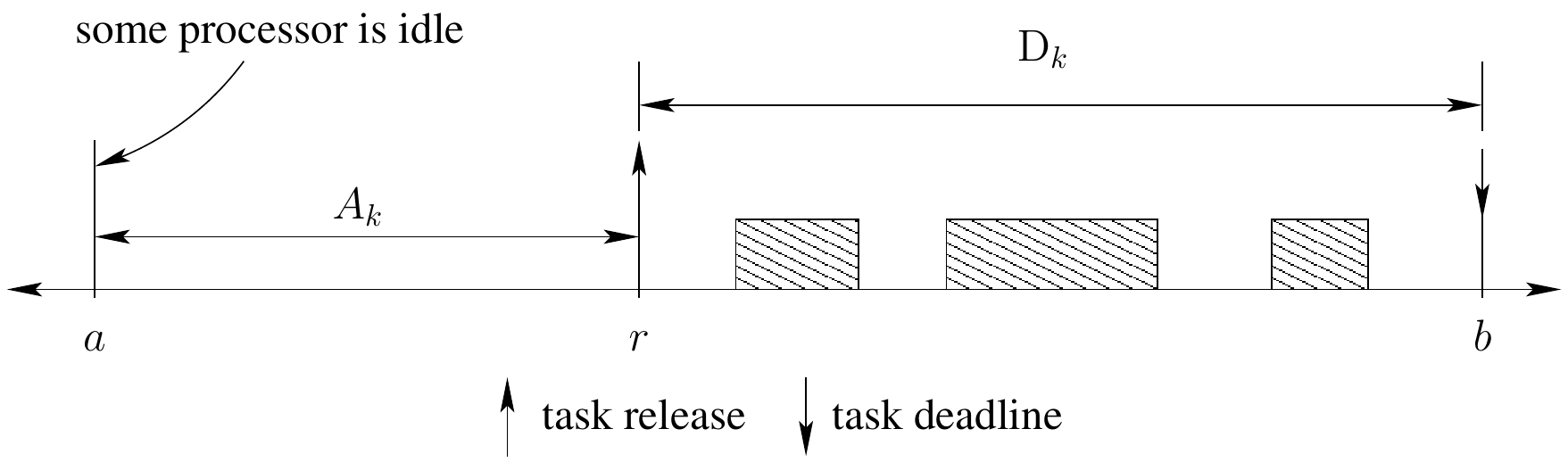}
\caption{Example time instant $a$ under dedicated resource}
\label{fig:mpr:baruah_schedulability} %\vspace{-1.0cm}
\end{figure}
When task $\task_k$ is scheduled on $\Rprocessors$ unit-capacity
processors under \GEDF, existing work identifies different time intervals that
must be checked to guarantee schedulability of
$\task_k$~\citep{Baruah07}. In particular, it assumes  
$b$ denotes the missed deadline of some job of task $\task_k$
(henceforth denoted as job $\task_k^b$), and then specifies different
values of $a$, corresponding to the interval $[a,b]$, that need to be
considered. Figure~\ref{fig:mpr:baruah_schedulability} gives one such 
time instant $a$. It corresponds to a point in time
such that: (1) at least one of the $\Rprocessors$ processors is idle
at that instant, (2) it is prior to the release time 
of job $\task_k^b$ ($r$ in the figure), and (3) no processor is idle
in the interval $(a,r]$. Observe that at each such time instant $a$ there
can be at most $\Rprocessors-1$ tasks that contribute towards carry-in
demand. This is because at most $\Rprocessors-1$ processors are
executing jobs at $a$. This observation is used to develop an
efficient schedulability condition in the dedicated resource
case. Informally, the study derives a condition on the total higher priority
workload in the interval $[a,b]$ that guarantees a deadline miss for
$\task_k^b$. In the following discussion we extend this notion of
time instant $a$ for the case when $\task_k$ is scheduled under the
partitioned resource supply $\mpr$.

When task $\task_k$ is scheduled using $\mpr$, we
denote a time instant as $t_{\mbox{idle}}$ if at least one of the $\Rprocessors$
processors is idle at that instant, even though it is available for
use as per supply $\mpr$. Figure~\ref{fig:mpr:idle_instant} illustrates one
such time instant, where $r$ denotes the release time of job
$\task_k^b$, $A_k$ denotes the length of the interval $(a,r]$ and $A_k +
D_k$ denotes the length of the interval $(a,b]$. 
%The empty boxes in
%that figure denote processor supply based on $\mpr$, and the striped
%boxes denote supply used by cluster $\component$. 
To check schedulability of task $\task_k$ we consider all time instants $a$
such that: (1) $a$ is $t_{\mbox{idle}}$, (2) $a \leq r$, and (3) no time
instant in the interval $(a,r]$ is $t_{\mbox{idle}}$. The time instant
illustrated in Figure~\ref{fig:mpr:idle_instant} satisfies these
properties.
\begin{figure}
\centering
\includegraphics[width=0.9\linewidth]{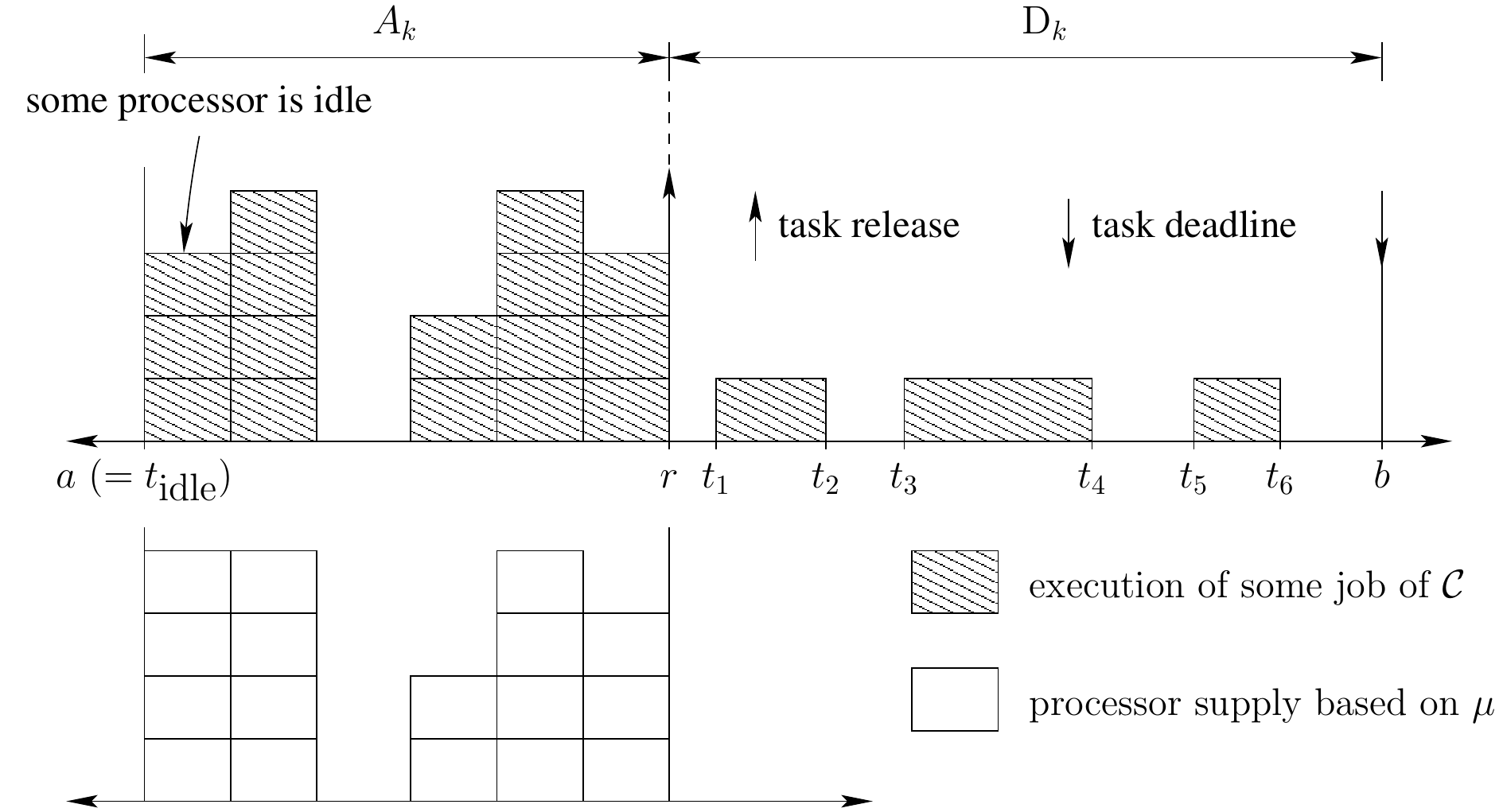}
\caption{Example time instant $t_{\mbox{idle}}$} %\vspace{-.3cm}
\label{fig:mpr:idle_instant}
\end{figure}

To derive the schedulability condition for component $\component$, we
consider all intervals $[a,b]$ as explained above and derive
conditions under which a deadline miss 
occurs for job $\task_k^b$. If $\task_k^b$ misses its deadline, then
the total workload of jobs having priority at least $\task_k^b$ must
be greater than the total processor supply available to $\component$
in $[a,b]$. Let $I_i$ ($1 \leq i \leq n$) denote the total
workload in interval $[a,b]$ of jobs of $\task_i$ that have priority
at least $\task_k^b$. Since $\sbf_{\mpr}(b-a)$ denotes a lower
bound on the processor supply available to $\component$ in $[a,b]$,
whenever $\task_k^b$ misses its deadline it must be true that 
\begin{equation}
\sum_{i = 1}^n I_i > \sbf_{\mpr}(b-a) =
\sbf_{\mpr}(A_k+\Tdeadline_k) \label{eqn:mpr:deadline_miss} 
\end{equation}
This inequality can be derived from the following observations: (1) the actual
processor supply available to component $\component$ 
in $[a,b]$ is at least $\sbf_{\mpr}(A_k+\Tdeadline_k)$ and
(2) there are no $t_{\mbox{idle}}$ time instants in the interval $(a,b]$,
\emph{i.e.}, all available processor supply is used by $\component$ to
schedule tasks from $\taskset$. For $\component$ to be schedulable using 
$\mpr$, it then suffices to show that for all tasks $\task_k$ and for
all values of $A_k$ Equation~\eqref{eqn:mpr:deadline_miss} is invalid.

We now derive an upper bound for each workload $I_i$. We separately
consider the workload of $\task_i$ in the following two interval
classes: (1) time intervals in $[a,b]$ in which $\task_k^b$
executes (intervals $[t_1,t_2], [t_3,t_4]$ and $[t_5,t_6]$ in
Figure~\ref{fig:mpr:idle_instant}) and (2) the other time intervals
in $[a,b]$. Let $I_{i,1}$ denote the workload of $\task_i$ in
intervals of type~(1) and $I_{i,2}$ denote the workload of $\task_i$ in
intervals of type~(2). We bound $I_i$ using upper bounds for $I_{i,1}$
and $I_{i,2}$. In the dedicated resource case, only intervals of type~(2)
were considered when deriving the schedulability
condition~\citep{Baruah07}. We however consider the contiguous interval
$[a,b]$, because $\sbf$ of MPR models are only defined over such
contiguous time intervals.

Since the cumulative length of intervals of type~(1) is at most $\Tcapacity_k$
and there are at most $\Rprocessors$ processors on which $\component$
executes, the total workload of all the tasks in intervals of 
type~(1) is clearly upper bounded by $\Rprocessors
\Tcapacity_k$. Therefore, $\sum_{i=1}^n I_{i,1} \leq \Rprocessors
\Tcapacity_k$. To bound 
$I_{i,2}$ we use the workload upper bound $\mathcal{W}_i$ presented in
Section~\ref{sec:mpr:higher_priority_interference}. Recall that
$\mathcal{W}_i(b-a)$ ($= \mathcal{W}_i(A_k+\Tdeadline_k)$) upper
bounds the workload of all jobs of $\task_i$ that execute in the
interval $[a,b]$ and have priority higher than $\task_k^b$. Therefore  
$\mathcal{W}_i(A_k+D_k)$ also upper bounds $I_{i,2}$. Further,
there is no need for $I_{i,2}$ to be larger than
$A_k+\Tdeadline_k-\Tcapacity_k$, because we have already considered a
total length of $\Tcapacity_k$ for intervals of type (1). Also this
bound can be further tightened for $i = k$, because in $I_{k,2}$ we
do not consider the executions of $\task_k^b$. These executions
are already considered for intervals of type~(1). Thus we can subtract
$\Tcapacity_k$ from $\mathcal{W}_k(A_k+\Tdeadline_k)$ and $I_{k,2}$
cannot be greater than $A_k$.
\begin{align*}
I_{i,2} \leq \bar{I}_{i,2} = \min \{ \mathcal{W}_i(A_k+\Tdeadline_k),
A_k+\Tdeadline_k-\Tcapacity_k \} & \mbox{ for all } i \not = k \\
I_{k,2} \leq \bar{I}_{k,2} = \min \{
\mathcal{W}_k(A_k+\Tdeadline_k)-\Tcapacity_k,A_k 
\} &
\end{align*}
Now by definition of time instant $a$ at most $\Rprocessors-1$ tasks can be
active, and hence have carry-in demand, at $a$. This follows from the
fact that at least one processor is not being used by 
$\component$ at $a$ even though that processor is
available as per supply $\mpr$. Hence we only need to
consider $\Rprocessors-1$ largest values of $CI_i$ when computing an
upper bound for $\sum_{i=1}^n I_{i,2}$ using the above equations, where
$CI_i$ denotes the carry-in demand in $\mathcal{W}_i$. Let us now
define the following two terms.
\begin{align*}
\hat{I}_{i,2} = \min \{ \mathcal{W}_i(A_k+\Tdeadline_k) -
 CI_i(A_k+\Tdeadline_k), A_k+\Tdeadline_k-\Tcapacity_k \} & \mbox{ for
   all } i \not = k \\
\hat{I}_{k,2} = \min \{ \mathcal{W}_k(A_k+\Tdeadline_k)-\Tcapacity_k -
CI_k(A_k+\Tdeadline_k), A_k \} &
\end{align*}
Let $L_{(\Rprocessors-1)}$ denote a set of task indices such that if $i
\in L_{(\Rprocessors-1)}$, then $(\bar{I}_{i,2} - \hat{I}_{i,2})$ is one
of the $\Rprocessors - 1$ largest values among all tasks. Then an 
upper bound on the worst-case resource demand in the interval $(a,b]$ can
be defined as, 
\begin{equation*}
\dem(A_k+\Tdeadline_k,\Rprocessors) =  \Rprocessors \Tcapacity_k + \sum_{i = 1}^n \hat{I}_{i,2} +
\sum_{i: i \in L_{(\Rprocessors-1)}} (\bar{I}_{i,2} - \hat{I}_{i,2}).
\end{equation*}
The following theorem gives our schedulability condition and its
proof follows from the above discussions.
\begin{theorem}
A component comprising of cluster $\component$ with $\Rprocessors$ processors
and sporadic tasks $\taskset = \{ \STask{1}, \ldots , \STask{n} \}$ 
is schedulable under \GEDF\ using MPR model $\GMpr$, if for all tasks
$\task_k \in \taskset$ and all $A_k \geq 0$,
\begin{equation}
\dem(A_k+\Tdeadline_k,\Rprocessors) \leq \sbf_{\mpr}(A_k+\Tdeadline_k). 
\label{eqn:mpr:EDF schedulability}
\end{equation}
\label{thm:mpr:EDF schedulability}
\end{theorem}

In Theorem~\ref{thm:mpr:EDF schedulability} if we set $\Rcapacity =
\Rprocessors \Rperiod$, then we get the schedulability condition under
dedicated resource that was proposed earlier~\citep{Baruah07}. This
shows that our condition is no more pessimistic than the one under
dedicated resource. Although this theorem gives a schedulability
test for component $\component$, it would be highly inefficient if
we were required to check for all values of $A_k$. The following
theorem shows that this is not the case. 
\begin{theorem}
If Equation~\eqref{eqn:mpr:EDF schedulability} is violated for some $A_k$,
then it must also be violated for a value satisfying the condition
\begin{equation*}
A_k < \frac{\Tcapacity_{\Sigma} + \Rprocessors \Tcapacity_k -
  \Tdeadline_k \left (\frac{\Rcapacity}{\Rperiod}-U_{\taskset} 
  \right) + U + B}{\frac{\Rcapacity}{\Rperiod} - U_{\taskset}},
\end{equation*}
where $\Tcapacity_{\Sigma}$ denotes the sum of $\Rprocessors-1$
largest $\Tcapacity_i$'s, $U_{\taskset} = \sum_{i = 1}^n
\frac{\Tcapacity_i}{\Tperiod_i}$, $U = \sum_{i = 1}^n
(\Tperiod_i-\Tdeadline_i) \frac{\Tcapacity_i}{\Tperiod_i}$ and $B =
\frac{\Rcapacity}{\Rperiod} \left [2 + 2\left (\Rperiod
    -\frac{\Rcapacity}{\Rprocessors} \right ) \right ]$.
\label{thm:mpr:Ak_bound}
\end{theorem}
\begin{proof}
It is easy to see that $\hat{I}_{i,2} \leq \dbf_{\task_i}(A_k+\Tdeadline_k)$ and
$\bar{I}_{i,2} \leq \dbf_{\task_i}(A_k+\Tdeadline_k) + \Tcapacity_i$,
where $\dbf_{\task_i}(t) = \left \lfloor \frac{t+T_i-D_i}{T_i}\right \rfloor
C_i$. Then the left hand side of Equation~\eqref{eqn:mpr:EDF
  schedulability} is less than or equal to
$\Tcapacity_{\Sigma}+\Rprocessors \Tcapacity_k+\sum_{i=1}^n
\dbf_{\task_i}(A_k+\Tdeadline_k)$. For this equation to be violated it
must be true that
\begin{align*}
& \Tcapacity_{\Sigma}+\Rprocessors \Tcapacity_k+\sum_{i=1}^n
\dbf_{\task_i}(A_k+\Tdeadline_k) > \sbf_{\mpr}(A_k+\Tdeadline_k) \\
\Rightarrow & (\mbox{Using } \dbf_{\task_i} \mbox{bound from~\citep{BMR90}}) \\
& \Tcapacity_{\Sigma}+\Rprocessors \Tcapacity_k+(A_k+\Tdeadline_k)U_{\taskset}+U >
\sbf_{\mpr}(A_k+\Tdeadline_k) \\ 
\Rightarrow & (\mbox{From Equation~\eqref{eqn:mpr:lsbf_MPR_2}}) \\
& \Tcapacity_{\Sigma}+\Rprocessors \Tcapacity_k+(A_k+\Tdeadline_k)U_{\taskset}+U>
\frac{\Rcapacity}{\Rperiod} \left (A_k+\Tdeadline_k \right )- B \\
\Rightarrow & (\mbox{Rearranging}) \\
& A_k < \frac{\Tcapacity_{\Sigma} + m \Tcapacity_k - \Tdeadline_k\left
    (\frac{\Rcapacity}{\Rperiod}-U_{\taskset} \right) + U +
  B}{\frac{\Rcapacity}{\Rperiod} - U_{\taskset}}
\end{align*}
\qed
\end{proof}

It can also be observed that Equation~\eqref{eqn:mpr:EDF schedulability}
only needs to be evaluated at those values of $A_k$ for which at
least one of $\hat{I}_{i,2},\bar{I}_{i,2}$ or $\sbf_{\mpr}$
change. Therefore Theorem~\ref{thm:mpr:EDF schedulability} gives a
pseudo-polynomial time schedulability condition whenever utilization
$U_{\taskset}$ is strictly less than the resource bandwidth
$\frac{\Rcapacity}{\Rperiod}$. In our techniques described later we
compute minimum possible $\Rcapacity$ and minimum required concurrency
$\Rprocessors$ for a given value of $\Rperiod$. Since $\Rcapacity$
appears inside floor and ceiling functions in $\sbf_{\mpr}$, these
computations may be intractable. We therefore replace $\sbf_{\mpr}$ in
Theorem~\ref{thm:mpr:EDF schedulability} with $\lsbf_{\mpr}$ from
Equation~\eqref{eqn:mpr:lsbf_MPR_2} before using it to
generate MPR interfaces.

\textbf{Discussion.} We have only focused on one intra-cluster
scheduling algorithm in this paper. However our
analysis technique can be easily extended to other intra-cluster
scheduling algorithms. Specifically, in the schedulability condition
given in Equation~\eqref{eqn:mpr:EDF schedulability},
$\dem(A_k+\Tdeadline_k,\Rprocessors)$ depends on \GEDF\ and
$\sbf_{\mpr}(A_k+\Tdeadline_k)$ depends on MPR model $\mpr$. Suppose
there exists a function $\dem_{\DM}(A_k+\Tdeadline_k,\Rprocessors)$
that can compute the workload upper bound for a task set scheduled 
under global \DM. Then we can plug in
$\dem_{\DM}(A_k+\Tdeadline_k,\Rprocessors)$ into
Equation~\eqref{eqn:mpr:EDF schedulability} to derive a schedulability
condition for global \DM\ intra-cluster scheduling. In fact, such
a $\dem_{\DM}$ can be obtained by extending current results over
dedicated resource~\citep{BCL05a}.

Bertogna and Cirinei have derived an upper bound for the
worst-case response time of tasks scheduled under \GEDF\ or global
\DM~\citep{BeCi07}. They have also used this bound to improve the
carry-in demand $CI_i$ that we use in our schedulability
condition. However this improvement to the carry-in demand cannot be
applied in our case. Since we schedule tasks using MPR model, any
response time computation depends on the processor supply in
addition to task demand. Then to use the response time bounds
presented in~\citep{BeCi07}, we must extend it with $\sbf$ of MPR 
model. However, since we are computing the MPR model (capacity
$\Rcapacity$ and concurrency $\Rprocessors$), its $\sbf$ is unknown
and therefore the response time is not computable. One way to resolve
this issue is to compute $\Rcapacity$ and $\Rprocessors$ using 
binary search. However, since $\Rcapacity$ belongs to the domain of
non-negative real numbers, binary
search for the minimum $\Rcapacity$ can take a prohibitively long time.

\section{Component interface generation}
\label{sec:mpr:interface_generation}

In this section we develop a technique to generate interface
$\GMpr$ for a cluster $\component$ comprising of sporadic tasks 
$\taskset = \{ \STask{1}, \ldots , \STask{n} \}$ scheduled under \GEDF. For
this purpose we use the schedulability condition given by 
Theorem~\ref{thm:mpr:EDF schedulability}. We assume that period $\Rperiod$ of
interface $\mpr$ is specified a priori by the system designer. For
instance, one can specify this period taking into account preemption
overheads in the system. We then compute values for capacity
$\Rcapacity$ and number of processors $\Rprocessors$ so that resource 
bandwidth of the interface is minimized. Finally, we also develop a
technique that transforms MPR interfaces to periodic
tasks\footnote{A periodic task $\GSTask$ is a special case of
  the identically defined sporadic task; $\Tperiod$ in the periodic
  case denotes the exact separation between successive job releases
  instead of minimum separation.}, in order to schedule clusters on
the multiprocessor platform (inter-cluster scheduling).

\subsection{Minimum bandwidth interface}
\label{sec:mpr:MEDP_optimality}

It is desirable to minimize the resource bandwidth of $\mpr$
when generating an interface for $\component$, because $\component$ then
consumes the minimum possible processor supply. We now give a lemma
which states that the resource bandwidth required to guarantee
schedulability of task set $\taskset$ monotonically increases as
number of processors in the cluster increases.
\begin{lemma} 
\label{lemma:mpr:min-cpu}
Consider interfaces $\Mpr{1}$ and $\Mpr{2}$, such that $\Rperiod_1
= \Rperiod_2$ and $\Rprocessors_2 = \Rprocessors_1 + 1$. Suppose these
two interfaces guarantee schedulability of the same component
$\component$ with their smallest possible resource bandwidth,
respectively. Then $\mpr_2$ has a higher resource bandwidth than
$\mpr_1$ does, \emph{i.e.}, $\Rcapacity_1 < \Rcapacity_2$.
\end{lemma}
\begin{proof}
We prove this lemma by contradiction. Consider $\mpr'_2 =
\tuple{\Rperiod_2,\Rcapacity'_2,\Rprocessors_2}$ such that
$\Rcapacity'_2 \leq \Rcapacity_1$. Suppose $\mpr'_2$ guarantees
schedulability of component $\component$ as per
Theorem~\ref{thm:mpr:EDF schedulability}. 

Let $\delta_d$ denote the difference in processor requirements of
$\component$ on $\Rprocessors_1$ and $\Rprocessors_2$ processors
for some interval length $A_k + D_k$, \emph{i.e.}, difference in 
function $\dem$ used in Theorem~\ref{thm:mpr:EDF schedulability}. Then 
\begin{align}
\delta_d = & \dem(A_k+\Tdeadline_k, \Rprocessors_2) -
\dem(A_k+\Tdeadline_k, \Rprocessors_1) \nonumber \\ 
= & \sum_{i: i \in L_{(\Rprocessors_2-1)}} (\bar{I}_{i,2} - \hat{I}_{i,2})
- \sum_{i: i \in L_{(\Rprocessors_1-1)}} (\bar{I}_{i,2} -
\hat{I}_{i,2}) + (\Rprocessors_2 - \Rprocessors_1) \Tcapacity_k
\nonumber \\ 
= & \sum_{i: i \in L_{(\Rprocessors_2-1)}} (\bar{I}_{i,2} - \hat{I}_{i,2})
- \sum_{i: i \in L_{(\Rprocessors_1-1)}} (\bar{I}_{i,2} -
\hat{I}_{i,2}) + \Tcapacity_k
\nonumber \\ 
> &  0 ~.
\label{eqn:mpr:delta-dbf}
\end{align}
It is indicated by $\delta_d > 0$ that the same component has a
greater upper bound on processor demand when it executes on more
processors. Now let $\delta_s$ denote the difference in the linear supply
bound function between $\mpr_1$ and $\mpr'_2$ for interval length $A_k
+ D_k$, \emph{i.e.}, 
\begin{align}
\delta_s = & \lsbf_{\mpr'_2}(A_k+\Tdeadline_k) -
\lsbf_{\mpr_1}(A_k+\Tdeadline_k) \nonumber \\
= & \frac{\Rcapacity'_2}{\Rperiod} \left( t - 2 \left ( \Rperiod + 1 -
  \frac{\Rcapacity'_2}{\Rprocessors_2} \right ) \right) -
\frac{\Rcapacity_1}{\Rperiod} \left( t - 2 \left ( \Rperiod + 1 -  
  \frac{\Rcapacity_1}{\Rprocessors_1} \right ) \right) \nonumber \\
\leq & \frac{\Rcapacity_1}{\Rperiod} \left( t - 2 \left ( \Rperiod + 1 -
  \frac{\Rcapacity_1}{\Rprocessors_2} \right ) \right) -
\frac{\Rcapacity_1}{\Rperiod} \left( t - 2 \left ( \Rperiod + 1 -  
  \frac{\Rcapacity_1}{\Rprocessors_1} \right ) \right) \nonumber \\
\leq & \frac{2({\Rcapacity_1})^2}{\Rperiod_1}
\left(\frac{1}{\Rprocessors_2} - \frac{1}{\Rprocessors_1} \right)
\nonumber \\ 
= & - \frac{2({\Rcapacity_1})^2}{\Rprocessors_1 \Rprocessors_2 \Rperiod_1}
\nonumber \\ 
< & 0 ~.
\label{eqn:mpr:delta-sbf}
\end{align}
It is indicated by $\delta_s < 0$ that MPR models provide less
processor supply with more available processors, when values of period 
and capacity are fixed. Thus $\delta_d > 0$ and $\delta_s < 0$ 
for all $A_k + D_k$. Since $\mpr_1$ guarantees schedulability of
component $\component$ using the smallest possible resource bandwidth,
$\dem(A_k+D_k,\Rprocessors_1) = \lsbf_{\mpr_1}(A_k+D_k)$ for some
$A_k+D_k$. Then $\dem(A_k+D_k,\Rprocessors_2) >
\lsbf_{\mpr'_2}(A_k+D_k)$ for that $A_k + D_k$, and therefore $\mpr'_2$
does not guarantee schedulability of $\component$ according to
Theorem~\ref{thm:mpr:EDF schedulability}. This contradicts the
assumption $\Rcapacity'_2 \leq \Rcapacity_1$.
\qed 
\end{proof}

Lemma~\ref{lemma:mpr:min-cpu} suggests that when we generate interface
$\mpr$, we should use the smallest number of processors to minimize resource
bandwidth of $\mpr$. However an arbitrarily small number for
$\Rprocessors$, say $\Rprocessors = 1$, may result in an
\emph{infeasible} $\mpr$. Recall 
that a MPR model $\GMpr$ is defined to be feasible if and only if
$\Rcapacity \leq \Rprocessors \Rperiod$. Therefore we find a feasible
interface $\mpr$ for $\component$ that: (1) guarantees schedulability
of $\component$ based on Theorem~\ref{thm:mpr:EDF schedulability} and
(2) uses the smallest possible number of processors ($m^*$). We can find such
$m^*$ through search. Since bandwidth is monotonic with number of
processors, a binary search can be performed to determine $m^*$. For
this search to terminate a lower and upper bound on $m^*$ should
be known. $\lceil U_{\taskset} \rceil$ is clearly a lower bound on the
number of processors necessary to schedule $\component$ where
$U_{\taskset} = \sum_i \frac{\Tcapacity_i}{\Tperiod_i}$. 
If the number of processors on the multiprocessor platform is known,
then that number can be used as an upper bound for $m^*$. Otherwise,
the following lemma gives an upper bound for $m^*$ as a function of
task parameters.

\begin{lemma}
If $\Rprocessors \geq \frac{\sum_{i=1}^n \Tcapacity_i}{\min_{i=1,
    \ldots , n} \{
  \Tdeadline_i - \Tcapacity_i \}} + n$, then feasible MPR model $\mpr
= \tuple{\Rperiod,\Rprocessors\Rperiod,\Rprocessors}$ guarantees schedulability of 
$\component$ as per Theorem~\ref{thm:mpr:EDF schedulability}.
\label{lem:mpr:upper bound m}
\end{lemma}
\begin{proof}
\begin{align}
\nonumber & \Rprocessors \geq \frac{\sum_{i=1}^n
  \Tcapacity_i}{\min_{i=1, \ldots , n} \{ \Tdeadline_i - \Tcapacity_i
  \}} + n & 
\\
\nonumber & (\mbox{Since } \forall k, A_k \geq 0 \mbox{ in Theorem~\ref{thm:mpr:EDF
    schedulability}}) & \\ 
\nonumber \Rightarrow & \Rprocessors \geq \frac{\sum_{i=1}^n
  \Tcapacity_i}{A_k+\Tdeadline_k-\Tcapacity_k} + n & 
\forall k \mbox{ and } \forall A_k \\
\Rightarrow & \Rprocessors (A_k+\Tdeadline_k-\Tcapacity_k) \geq \sum_{i=1}^n
\Tcapacity_i + n(A_k+\Tdeadline_k-\Tcapacity_k) & 
\forall k \mbox{ and } \forall A_k
\label{eqn:mpr:upper bound m}
\end{align}
Now consider the function $\dem(A_k+\Tdeadline_k,\Rprocessors)$ from
Theorem~\ref{thm:mpr:EDF schedulability}.
\begin{align*}
& \dem(A_k+\Tdeadline_k,\Rprocessors) =  \sum_{i = 1}^n \hat{I}_{i,2}
+ \sum_{i: i \in L_{(\Rprocessors-1)}} (\bar{I}_{i,2} -
\hat{I}_{i,2}) + \Rprocessors \Tcapacity_k \\
& \left (\mbox{Since each } \hat{I}_{i,2} \leq A_k+\Tdeadline_k-\Tcapacity_k 
\mbox{ and } \sum_{i: i \in L_{(\Rprocessors-1)}} (\bar{I}_{i,2} -
\hat{I}_{i,2}) \leq \sum_{i=1}^n \Tcapacity_i \right ) \\
\Rightarrow   & \dem(A_k+\Tdeadline_k,\Rprocessors) \leq
              n(A_k+\Tdeadline_k-\Tcapacity_k) + \sum_{i=1}^n 
                \Tcapacity_i + \Rprocessors \Tcapacity_k \\ 
& (\mbox{From Equation~\eqref{eqn:mpr:upper bound m}}) \\
\Rightarrow   &  \dem(A_k+\Tdeadline_k,\Rprocessors) \leq
              \Rprocessors(A_k+\Tdeadline_k-\Tcapacity_k) + 
                \Rprocessors \Tcapacity_k \\
\Rightarrow   & \dem(A_k+\Tdeadline_k,\Rprocessors) \leq
\sbf_{\mpr}(A_k+\Tdeadline_k)
\end{align*}
Since this inequality holds for all $k$ and $A_k$, from
Theorem~\ref{thm:mpr:EDF schedulability} we get that $\mpr$ is
guaranteed to schedule $\component$.
\qed
\end{proof}

Since $\mpr$ in Lemma~\ref{lem:mpr:upper bound m} is feasible and
guarantees schedulability of $\component$, $\frac{\sum_{i=1}^n
  \Tcapacity_i}{\min_{i=1}^n \{ \Tdeadline_i - \Tcapacity_i \}} + n$
is an upper bound for $m^*$. Thus we generate an interface for
$\component$ by doing a binary search for $m^*$ in the range $\left [ \lceil
U_{\taskset} \rceil, \frac{\sum_{i=1}^n \Tcapacity_i}{\min_{i=1}^n \{
  \Tdeadline_i - \Tcapacity_i \}} + n \right ]$. For each value of the number of
processors $\Rprocessors$, we compute the smallest value of
$\Rcapacity$ that satisfies Equation~\eqref{eqn:mpr:EDF
  schedulability} in Theorem~\ref{thm:mpr:EDF
  schedulability}, assuming $\sbf_{\mpr}$ is replaced with
$\lsbf_{\mpr}$. $\Rcapacity$ (= $\Rcapacity^*$), corresponding to the
smallest value of $\Rprocessors$ (= $m^*$) that 
guarantees schedulability of $\component$ and results in a feasible
interface, is then chosen as the capacity of $\mpr$. Also $m^*$ is
chosen as the number of processors in the cluster,
\emph{i.e.}, $\mpr = \tuple{\Rperiod,\Rcapacity^*,m^*}$. 

\textbf{Algorithm complexity.} To bound $A_k$ as in
Theorem~\ref{thm:mpr:Ak_bound} we must know the value of
$\Rcapacity$. However, since $\Rcapacity$ is being computed, we use
its smallest ($0$) and largest ($\Rprocessors \Rperiod$) possible 
values to bound $A_k$. For each value of $\Rprocessors > U_{\taskset}$,
$\Rcapacity$ can then be computed in pseudo-polynomial time using
Theorem~\eqref{thm:mpr:EDF schedulability}, assuming $\sbf$ is replaced
with $\lsbf$. This follows from the fact that the denominator in the
bound of $A_k$ in Theorem~\ref{thm:mpr:Ak_bound} is non-zero. The only
problem case is when $\Rprocessors = \lceil U_{\taskset} \rceil =
U_{\taskset}$. However in this case, we now show that $\mpr =
\tuple{\Rperiod,\Rprocessors\Rperiod,\Rprocessors}$ can schedule
$\component$ if and only if, $\Rprocessors=1$ and $\Tdeadline_i \geq
\Tperiod_i$ for each task $\task_i$ in $\taskset$. Clearly, if some
$\Tdeadline_i < \Tperiod_i$, then a resource bandwidth of $U_{\taskset}$ is
not sufficient to guarantee schedulability. Now suppose $\Rprocessors
> 1$. Then the left hand side of Equation~\eqref{eqn:mpr:EDF
  schedulability} is $> \sum_{i=1}^n \dbf_{\task_i}(A_k+\Tdeadline_k) \geq (A_k+\Tdeadline_k) U_{\taskset}
$, because $\hat{I}_{i,2} = \dbf_{\task_i}(A_k+\Tdeadline_k)$,
$\bar{I}_{i,2} \geq \dbf_{\task_i}(A_k+\Tdeadline_k)$, and
$\Rprocessors \Tcapacity_k > 0$. Hence in this case $\Rprocessors >
U_{\taskset}$ and this is a contradiction. Therefore computing the interface
for $\Rprocessors = U_{\taskset}$ can be done in constant time. The
number of different values of $\Rprocessors$ to be considered 
is polynomial in the input size,
because the search interval is bounded by numbers that are polynomial
in the input parameters. Therefore the entire interface generation
process has pseudo-polynomial complexity.

\begin{table*}[tbh]
\begin{center}
\footnotesize
\begin{tabular}{||c|c|c|c||}
\hline
      & & & \\
\multicolumn{1}{||c|}{Cluster} &
\multicolumn{1}{|c|}{Task set} &
\multicolumn{1}{|c|}{$\sum_i \frac{\Tcapacity_i}{\Tperiod_i}$} &
\multicolumn{1}{|c||}{$\sum_i \frac{\Tcapacity_i}{\Tdeadline_i}$} \\
      & & & \\
\hline
      & & & \\
$\component_1$ & $\{ (60,5,60), (60,5,60), (60,5,60), (60,5,60), (70,5,70), (70,5,70),$ & $1.304$ & $1.304$ \\
      & $(80,5,80), (80,5,80), (80,10,80), (90,5,90), (90,10,90),$ & & \\
      & $(90,10,90), (100,10,100), (100,10,100), (100,10,100) \}$ & & \\
      & & & \\
\hline
     & & & \\
$\component_2$ & $\{ (60,5,60), (100,5,100) \}$ & $0.1333$ & $0.1333$ \\
      & & & \\
\hline
      & & & \\
$\component_3$ & $\{ (45,2,40), (45,2,45), (45,3,40), (45,3,45), (50,5,45),$ & $1.1222$ & $1.1930$ \\
      & $(50,5,50), (50,5,50), (50,5,50), (70,5,60),  (70,5,60),$ & & \\
      & $(70,5,65), (70,5,65), (70,5,65), (70,5,65), (70,5,70) \}$  & & \\
      & & & \\
\hline
\end{tabular}
\normalsize
\end{center}
\caption{Clusters $\component_1, \component_2$ and $\component_3$}
\label{tab:mpr:components}
\end{table*}
\begin{figure*}
\centering
\subfigure[$\mpr_1^*$]{
\includegraphics[width=0.47\linewidth]{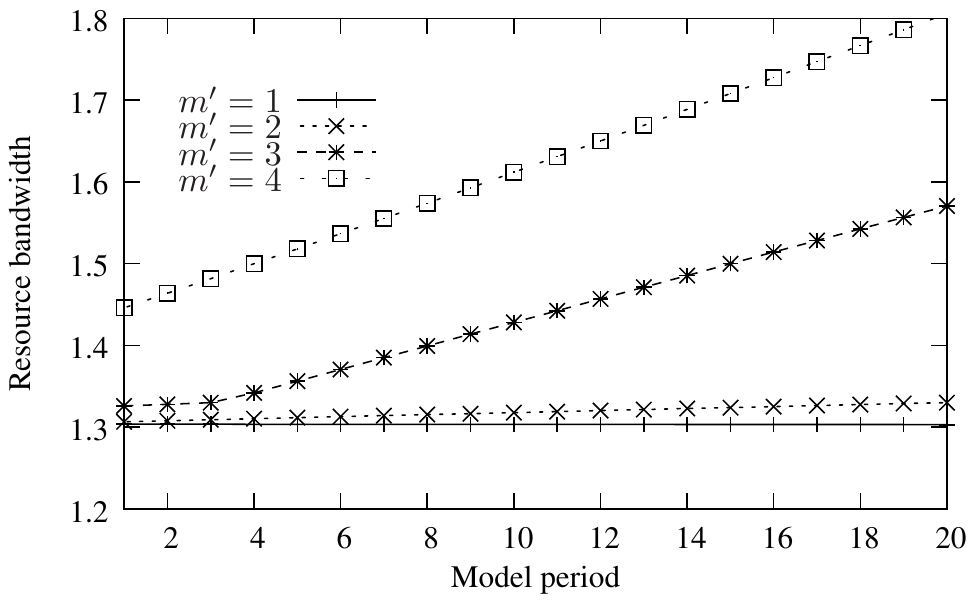}
\label{fig:mpr:I1}
}
\subfigure[$\mpr_2^*$]{
\includegraphics[width=0.47\linewidth]{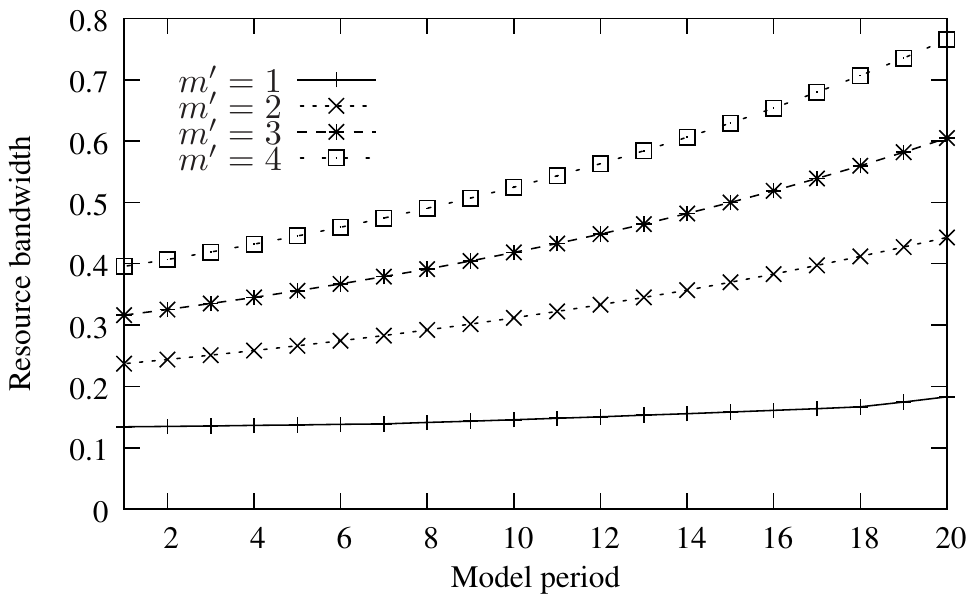}
\label{fig:mpr:I2}
}
\subfigure[$\mpr_3^*$]{
\includegraphics[width=0.47\linewidth]{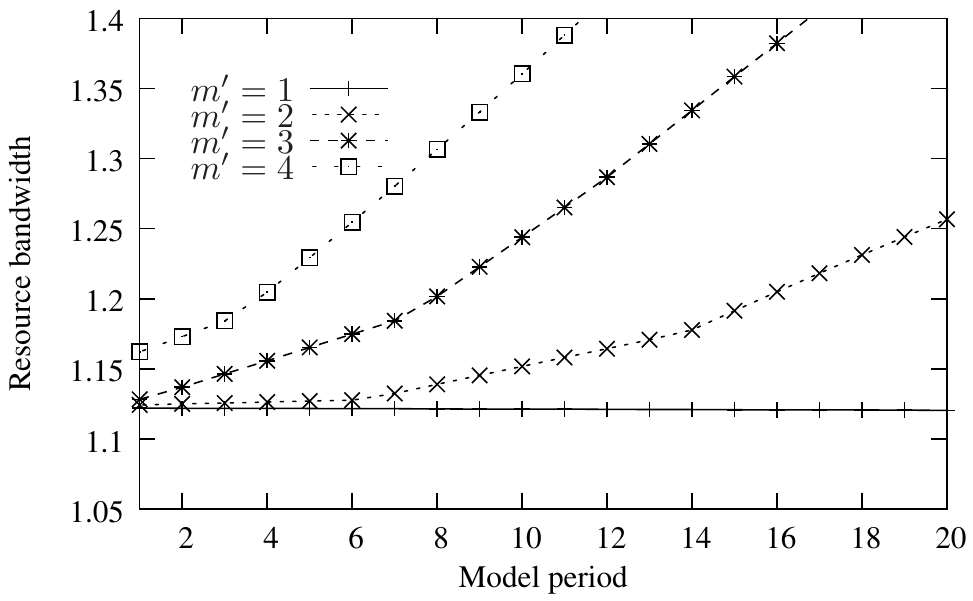}
\label{fig:mpr:I3} 
} 
\subfigure[$\mpr^*$]{
\includegraphics[width=0.47\linewidth]{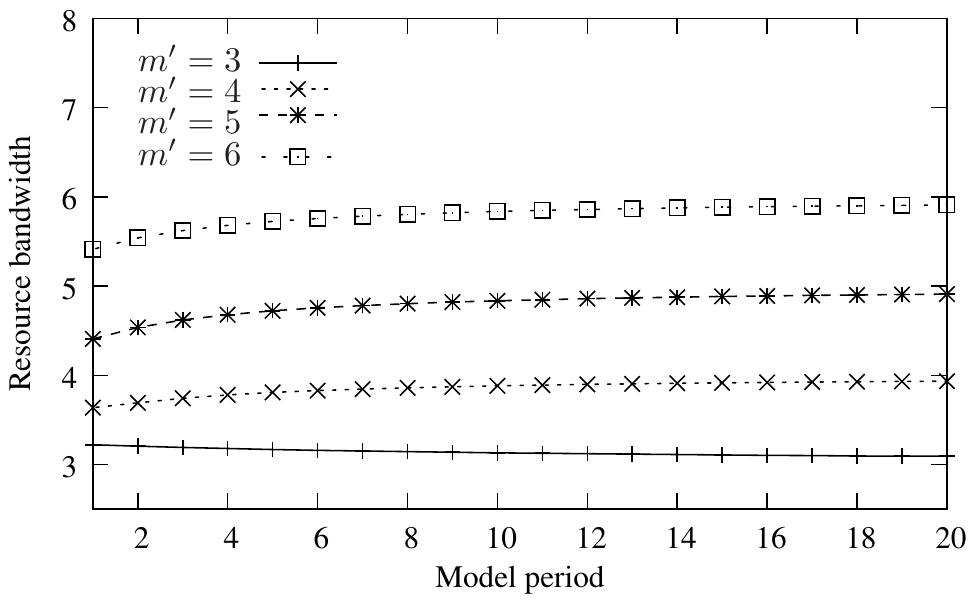}
\label{fig:mpr:I4} 
}
\caption{MPR model based interfaces} 
\label{fig:mpr:I}
\end{figure*}

\begin{example}
Consider the example virtual clustering framework 
shown in Figure~\ref{fig:mpr:hierarchy_components}. Let clusters
$\component_1, \component_2$ and $\component_3$ be assigned tasks
as shown in Table~\ref{tab:mpr:components}. Interfaces $\mpr_1^*,
\mpr_2^*$ and $\mpr_3^*$, for clusters $\component_1,
\component_2$ and $\component_3$, are shown in
Figures~\ref{fig:mpr:I1},~\ref{fig:mpr:I2} and~\ref{fig:mpr:I3} respectively.
In the figures we have plotted the resource bandwidth of these interfaces
for varying periods and $\Rprocessors$ denotes the number of processors in
the cluster. 
\label{eg:mpr:example1}
\end{example}

Figures~\ref{fig:mpr:I1} and~\ref{fig:mpr:I3} show that when
$\Rprocessors=1$ 
interfaces $\mpr_1^*$ and $\mpr_3^*$ are not feasible; their resource 
bandwidths are greater than $1$ for all period 
values. This shows that clusters $\component_1$ and $\component_3$ are
not schedulable on clusters having one processor. This is as expected because 
the utilization of task sets in these clusters is also greater than one.
However when $\Rprocessors=2$, $\mpr_1^*$ and $\mpr_3^*$ are feasible,
\emph{i.e.}, their respective resource bandwidths are at most
two. Therefore for clusters $\component_1$ and $\component_3$, we
choose MPR interfaces $\mpr_1^*$ and $\mpr_3^*$ with $\Rprocessors=2$. Similarly, Figure~\ref{fig:mpr:I2} shows 
that $\mpr_2^*$ is a feasible interface for cluster $\component_2$ when $\Rprocessors=1$.
These plots also show that resource overheads\footnote{Difference
  between $\max_k \max_{A_k}
  \frac{\dem(A_k+D_k,\Rprocessors)}{A_k+D_k}$ and resource bandwidth
  of MPR interface.}
incurred by our interfaces are small for the non-trivial examples
presented here.

\subsection{Inter-cluster scheduling}

As discussed in the introduction, virtual clustering 
involves two-level scheduling; scheduling of tasks within each
cluster (intra-cluster scheduling) and scheduling of clusters on the
multiprocessor platform (inter-cluster scheduling). MPR interfaces
generated in the previous section capture task-level concurrency
constraints within a cluster. Hence inter-cluster scheduling need
not worry about these constraints when it schedules cluster 
interfaces. However there is no known scheduling 
algorithm for MPR interfaces. Therefore we now develop
a technique to transform a MPR model into periodic tasks such that
processor requirements of these tasks are at least as much as those of
the resource model.
\begin{definition}
Consider a MPR model $\mpr = \tuple{\Rperiod,\Rcapacity^*,m^*}$ and
let $\alpha = \Rcapacity^* - m^* \left \lfloor
  \frac{\Rcapacity^*}{m^*} \right \rfloor$ and $k = \lfloor \alpha
\rfloor$. Define the transformation from $\mpr$ to a periodic task set
$\taskset_{\mpr}$ as \\
$\taskset_{\mpr} = \{ \STask{1}, \ldots , \STask{m^*} \}$, where \\
$\task_1 = \ldots = \task_k = \left (\Rperiod,\left \lfloor
    \frac{\Rcapacity^*}{m^*} \right \rfloor + 1,\Rperiod \right )$, \\
$\task_{k + 1} = \left ( \Rperiod, \left \lfloor \frac{\Rcapacity^*}{m^*} \right \rfloor
  + \alpha - k \left \lfloor \frac{\alpha}{k} \right \rfloor, \Rperiod \right
)$ and \\
$\task_{k + 2} = \ldots = \task_{m^*} = \left (\Rperiod,\left \lfloor
    \frac{\Rcapacity^*}{m^*} \right \rfloor,\Rperiod \right)$.
\label{def:mpr:demand-supply optimal}
\end{definition}
In this definition it is easy to see that the total processor demand
of $\taskset_{\mpr}$ is $\Rcapacity^*$ in every period
$\Rperiod$. Further, we have assumed that whenever $\Rcapacity^*$ is
not an integer, processor supply from $\mpr$ fully utilizes one
processor before using another. For example, if
$\Rcapacity^* = 2.5$ and $m^* = 3$, then $\mpr$ will provide two units
of resource from two processors and 
the remaining $0.5$ units from the third processor. The following theorem
proves correctness of this transformation.
\begin{theorem}
If all the deadlines of task set $\taskset_{\mpr}$ in
Definition~\ref{def:mpr:demand-supply optimal} are met by some
processor supply with concurrency at most $m^*$ at any time instant,
then its supply bound function is lower bounded by $\sbf_{\mpr}$.
\label{thm:mpr:demand_supply optimal}
\end{theorem}
\begin{proof}
Since $\taskset_{\mpr}$ has $m^*$ tasks, it can utilize at most $m^*$
processors at any time instant. Therefore if some processor supply
provides more than $m^*$ processors at any time instant, then we can
ignore these additional processor allocations. Hence we only need to
consider processor supplies with concurrency at most $m^*$. 

Total processor demand of all the tasks in $\taskset_{\mpr}$ is
$\Rcapacity^*$ in every period of $\Rperiod$ time units. Then to meet
all the deadlines of task set $\taskset_{\mpr}$, any processor supply
must provide at least $\Rcapacity^*$ processor units in every period of
$\Rperiod$ time units, with amount of concurrency at most $m^*$. But this is
exactly the definition of MPR model $\mpr =
\tuple{\Rperiod,\Rcapacity^*,m^*}$. Therefore the supply bound
function of this processor supply is lower bounded by $\sbf_{\mpr}$.  
\qed
\end{proof}

Thus MPR interfaces generated in the previous section can be
transformed into periodic tasks using
Definition~\ref{def:mpr:demand-supply optimal}. Once such tasks are
generated for each virtual cluster, inter-cluster scheduling can be
done using existing multiprocessor algorithms like \GEDF,
Pfair~\citep{BCP96}, etc.
\begin{example}
For MPR interfaces $\mpr_1^*,\mpr_2^*$ and $\mpr_3^*$
generated in Example~\ref{eg:mpr:example1}, we select periods
$6,8,$ and $5$ respectively, \emph{i.e.}, interfaces
$\tuple{6,8.22,2}, \tuple{8,2.34,1}$ and $\tuple{5,5.83,2}$. Using 
Definition~\ref{def:mpr:demand-supply optimal} we 
get task sets $\taskset_{\mpr_1^*} = \{ (6,5,6), (6,4,6)
\}, \taskset_{\mpr_2^*} = \{ (8,3,8) \}$ and $\taskset_{\mpr_3^*} = \{
(5,3,5), (5,3,5) \}$. Suppose the three clusters
$\component_1,\component_2$ and $\component_3$ (\emph{i.e.}, task set
$\{ \taskset_{\mpr_1^*}, \taskset_{\mpr_2^*}, \taskset_{\mpr_3^*} \}$) are
scheduled on a multiprocessor platform using \GEDF. Then the
resulting MPR interface $\mpr^*$ is plotted in Figure~\ref{fig:mpr:I4}.
As shown in the figure, $\mpr^*$ is not feasible for
$\Rprocessors=3$; its resource bandwidth is greater than $3$ for all
period values. However these three clusters are schedulable on a
multiprocessor platform having $4$ processors (in the figure $\mpr^*$
is feasible when $\Rprocessors=4$).
\label{eg:mpr:example2}
\end{example}

The above example clearly illustrates the advantage of virtual
clustering over physical clustering. The three components
$\component_1$, $\component_2$ and $\component_3$, would require $5$
processors under physical clustering ($2$ each for $\component_1$ and
$\component_3$ and $1$ for $\component_2$). On the other hand, a
\GEDF\ based virtual clustering technique can schedule these clusters
using only $4$ processors. Although total utilization of tasks in the
three clusters is $2.56$, our analysis requires $4$ processors to
schedule the system. This overhead is as a result of the
following factors: (1) \GEDF\ is not an optimal scheduling
algorithm on multiprocessor platforms (both for intra- and
inter-cluster scheduling), (2) the schedulability
conditions we use are only sufficient conditions, and (3) capturing
task-level concurrency constraints in a component interface
leads to some increase in processor requirements (resource overhead of
abstracting a cluster into MPR interface).

\section{Virtual cluster-based scheduling algorithms}
\label{sec:mpr:improved_virtual_cluster_main}

In this section we propose new virtual-cluster based scheduling algorithms
for implicit deadline sporadic task systems. Prior to presenting these
algorithms, we eliminate resource overheads from the virtual
clustering framework proposed in Section~\ref{sec:mpr:interface_generation}. 

\subsection{Improved virtual-clustering framework}
\label{sec:mpr:improved_virtual_cluster}

In this section we present an inter-cluster scheduling algorithm that is optimal
whenever all the MPR interfaces being scheduled under it have
identical periods. We also present another transformation from MPR
models to periodic tasks, which along with the optimal inter-cluster
scheduler, results in an improved $\sbf$ for MPR models. These two
together, eliminate the resource overheads described at the end of
previous section.

McNaughton~\citep{McN59} presented an algorithm for scheduling real-time jobs in a given
time interval on a multiprocessor platform. This algorithm can be explained as follows:
Consider $n$ jobs to be scheduled on $m$ processors in a time interval
$(t_1,t_2]$ of length $t$, such that no job is simultaneously
scheduled on more than one processor. The job set need not be sorted
in any particular order. McNaughton's algorithm schedules the $i^{th}$
job on the first non-empty processor, packing jobs from left to
right. Suppose the $(i-1)^{st}$ job was scheduled on processor $k$ up to
time instant $t_3$ ($t_1 \leq t_3 \leq t_2$). Then up to $t_2 - t_3$
time units of the $i^{th}$ job are scheduled on processor $k$ and the
remaining time units are scheduled on processor $k+1$ starting from
$t_1$. Figure~\ref{fig:mpr:mcnaughton_job} illustrates this schedule
for a job set $\{ J_1, \ldots , J_5 \}$ on $4$ processors. Note that
if the total resource demand of a job is at most $t_2 - t_1$, then (1) 
the job is scheduled on at most two processors by McNaughton's
algorithm, and (2) the job is never scheduled simultaneously on 
both the processors. The following theorem establishes conditions
under which this algorithm can successfully schedule job sets.
\begin{figure}
\centering
\includegraphics[width=0.6\linewidth]{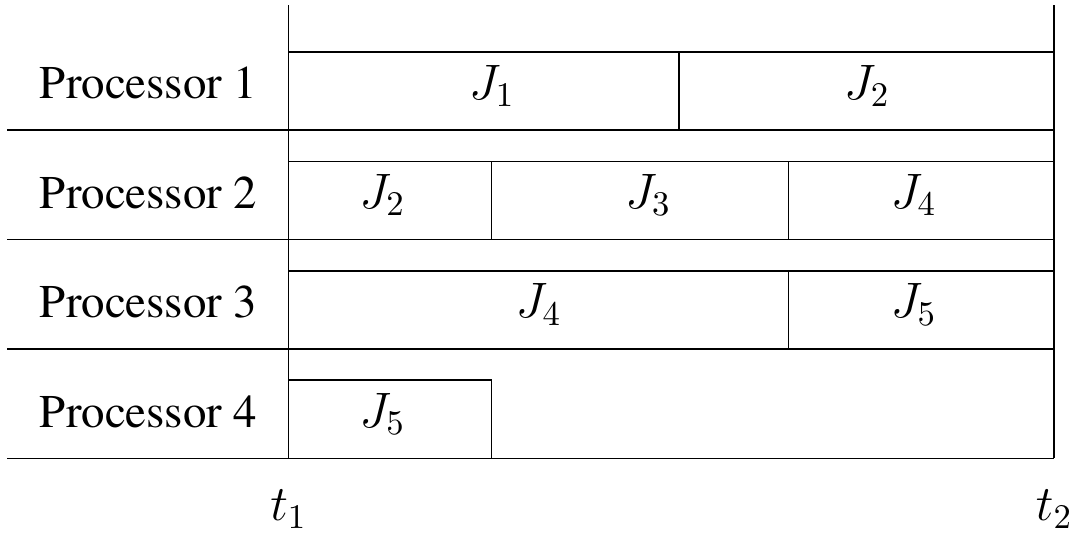}
\caption{Schedule of job set under McNaughton's algorithm in the interval
  $(t_1,t_2]$}
\label{fig:mpr:mcnaughton_job}
\end{figure}
\begin{theorem}[Theorem 3.1 in~\citep{McN59}]
Let $\Jcapacity_1, \ldots , \Jcapacity_n$ denote the number of
processor units of the $n$ jobs that must be scheduled in the interval
$(t_1, t_2]$ on $m$ identical, unit-capacity processors. If $\sum_{i=1}^n \Jcapacity_i
\leq m(t_2 - t_1)$, then a necessary and sufficient condition to
guarantee schedulability of this job set is that for all $i$,
$\Jcapacity_i \leq t_2 - t_1$.  
\label{thm:mpr:mcnaughton_optimal}
\end{theorem}

If $\Jcapacity_i > t_2 - t_1$ then the $i^{th}$ job cannot be
scheduled in the interval $(t_1, t_2]$ by any scheduling algorithm, unless
the job is simultaneously scheduled on more than one
processor. Likewise, if $\sum_{i=1}^n \Jcapacity_i > m(t_2 - t_1)$,
then also the job set cannot be scheduled by any scheduling algorithm,
because the total processor demand in the interval $(t_1, t_2]$ is greater
than the total available processing capacity. Hence 
Theorem~\ref{thm:mpr:mcnaughton_optimal} in fact shows that
McNaughton's algorithm is optimal for scheduling job sets in a given
time interval.

Consider a periodic task set $\taskset = \{ (\Tperiod, \Tcapacity_1,
\Tperiod), \ldots , (\Tperiod, \Tcapacity_n, \Tperiod) \}$. Tasks in
$\taskset$ have identical periods and implicit deadline. Suppose we
use McNaughton's algorithm in the intervals $(k \Tperiod, (k+1)
\Tperiod]$, $k \in \mathbb{I}$, to schedule jobs of $\taskset$ on $m$
identical, unit-capacity processors. Then for each interval $(k
\Tperiod, (k+1) \Tperiod]$ (1) all jobs of $\taskset$ are released at
the beginning of the interval ($k \Tperiod$) and (2) all jobs of
$\taskset$ have deadline at the end of the interval ($(k+1)
\Tperiod$). Therefore, from Theorem~\ref{thm:mpr:mcnaughton_optimal}, 
we get that McNaughton's algorithm optimally schedules these jobs in
each interval and this leads to the following direct corollary.
\begin{corollary}
Let $\taskset = \{ \STask{1}, \ldots , \STask{n} \}$ denote a periodic
task set to be scheduled on $m$ identical, unit-capacity
processors. If $\Tperiod_1 = \ldots = \Tperiod_n = \Tdeadline_1 =
\ldots = \Tdeadline_n$ $(= \Tperiod)$, then a necessary and sufficient
condition for $\taskset$ to be schedulable using McNaughton's
algorithm is that $\sum_{i = 1}^n \Tcapacity_i \leq m \Tperiod$ and
$\Tcapacity_i \leq \Tperiod_i$ for each $i$.
\label{cor:mpr:mcnaughton_optimal_MPR}
\end{corollary}

Consider the virtual clustering framework proposed in
Section~\ref{sec:mpr:interface_generation}. Suppose all the MPR
interfaces in this framework have identical periods. Then all the
periodic tasks generated using Definition~\ref{def:mpr:demand-supply
  optimal} also have identical periods. And from
Corollary~\ref{cor:mpr:mcnaughton_optimal_MPR} we get that McNaughton's
algorithm is optimal for scheduling these tasks on the physical
platform, \emph{i.e.}, the algorithm does not incur any resource
overhead for inter-cluster scheduling.

\begin{figure}
\centering
\includegraphics[width=0.95\linewidth]{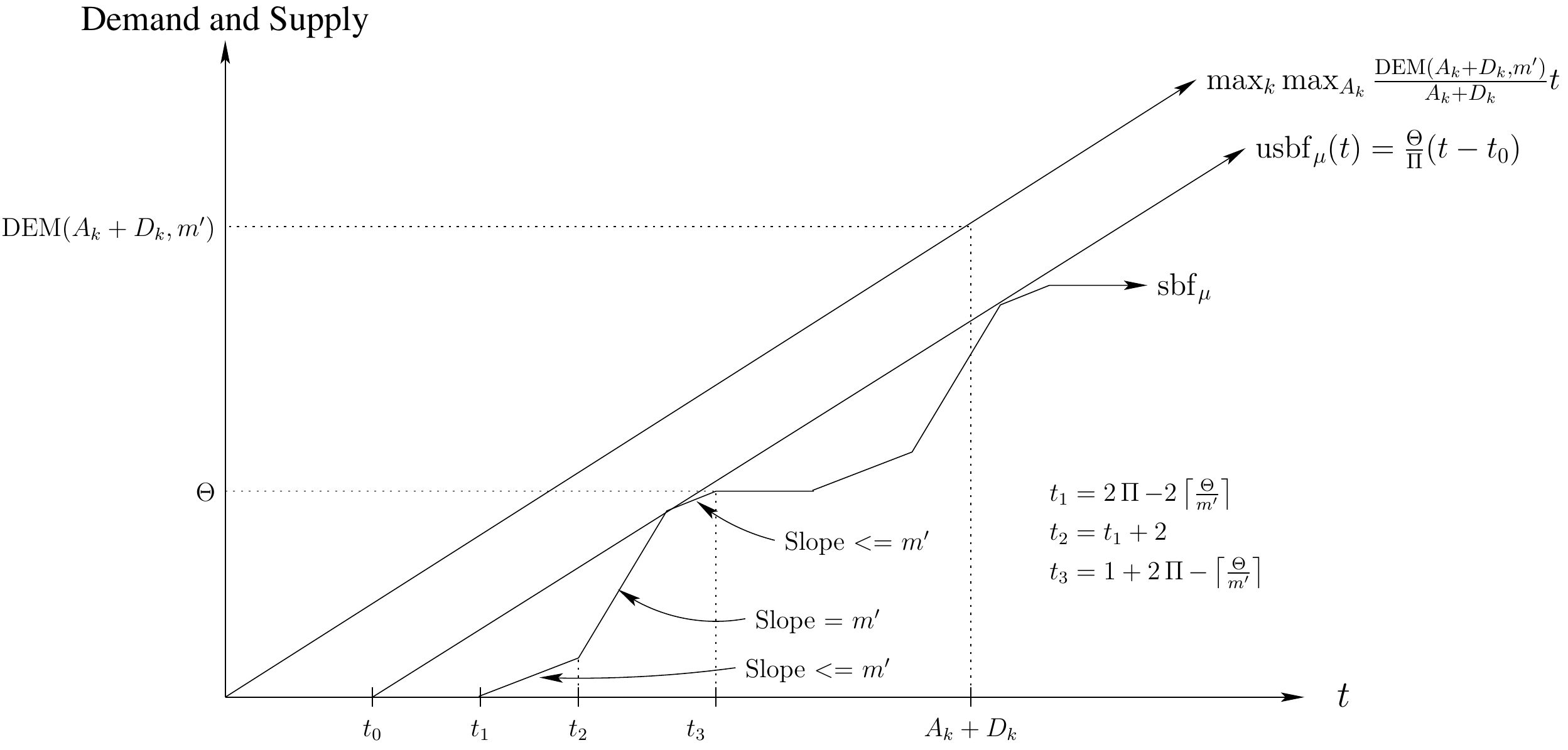}
\caption{Bandwidth of $\sbf_{\mpr}$ and schedulability load of cluster
  $\component$}
\label{fig:mpr:sbf_suboptimal}
\end{figure}
Another source of resource overhead is the abstraction of a
cluster into MPR interface and its transformation to a periodic task
set. This overhead results from the sub-optimality of $\sbf$ of MPR 
models which can be explained as follows. Consider the two functions,
$\sbf_{\mpr}$ and $\usbf_{\mpr}$, shown in
Figure~\ref{fig:mpr:sbf_suboptimal}. The resource bandwidth used by
$\mpr$ is equal to the slope of line $\usbf_{\mpr}$
($\frac{\Rcapacity}{\Rperiod}$). Suppose $\mpr$ is used to abstract
the processor demand of cluster $\component$ 
in Theorem~\ref{thm:mpr:EDF schedulability}. Since
$\sbf_{\mpr}$ has a non-zero x-axis intercept, the bandwidth
of $\mpr$ is strictly larger than the schedulability
load, $\max_k
\max_{A_k}\dem(A_k+\Tdeadline_k,\Rprocessors)/(A_k+\Tdeadline_k)$, of
cluster $\component$. If not then, as shown in
Figure~\ref{fig:mpr:sbf_suboptimal}, there exists some $A_k+D_k$ for
which Theorem~\ref{thm:mpr:EDF schedulability} is not satisfied. This
explains the resource overhead in the abstraction of clusters to MPR
interfaces. Now suppose $\mpr$ is transformed into the periodic task set 
$\taskset_{\mpr}$ using Definition~\ref{def:mpr:demand-supply
  optimal}. Then from Theorem~\ref{thm:mpr:demand_supply optimal} we
get that the total processor demand of $\taskset_{\mpr}$ is at least as
much as $\sbf_{\mpr}$. However, since $\sbf_{\mpr}$ does not guarantee
$\Rcapacity$ resource units in an interval of length $\Rperiod$ (see
Figure~\ref{fig:mpr:sbf_suboptimal}), a processor supply with supply
bound function exactly $\sbf_{\mpr}$ cannot schedule
$\taskset_{\mpr}$. This explains the resource overhead in the 
transformation of MPR interfaces to periodic tasks.

To eliminate the aforementioned overheads, we must modify the
transformation presented in Definition~\ref{def:mpr:demand-supply
  optimal}. This is because the schedule of $\taskset_{\mpr}$
determines the processor supply from the multiprocessor
platform to $\mpr$, and this in turn determines
$\sbf_{\mpr}$. We now present a new transformation from MPR models to
periodic tasks as follows. 
\begin{definition}
Given a MPR model $\mpr = \tuple{\Rperiod,\Rcapacity^*,m^*}$, we
define its transformation to a periodic task set $\taskset_{\mpr}$ as
\\
$\taskset_{\mpr} = \{ \STask{1}, \ldots , \STask{m^*} \}$, where \\
$\task_1 = \ldots = \task_{m^*-1} = (\Rperiod, \Rperiod, \Rperiod)$
and \\
$\task_{m^*} = (\Rperiod, \Rcapacity^* - (m^*-1)\Rperiod,
\Rperiod)$.
\label{def:mpr:demand-supply optimal_improved}
\end{definition}
In this definition it is easy to see that the total processor demand
of $\taskset_{\mpr}$ is $\Rcapacity^*$ in every $\Rperiod$ 
time units, with concurrency at most $m^*$. Therefore 
Theorem~\ref{thm:mpr:demand_supply optimal} holds in this case as
well, \emph{i.e.}, if all the deadlines of task set $\taskset_{\mpr}$
are met by some processor supply with concurrency at most $m^*$ at any
time instant, then its supply bound function is lower bounded by
$\sbf_{\mpr}$.

\begin{figure}
\centering
\includegraphics[width=0.85\linewidth]{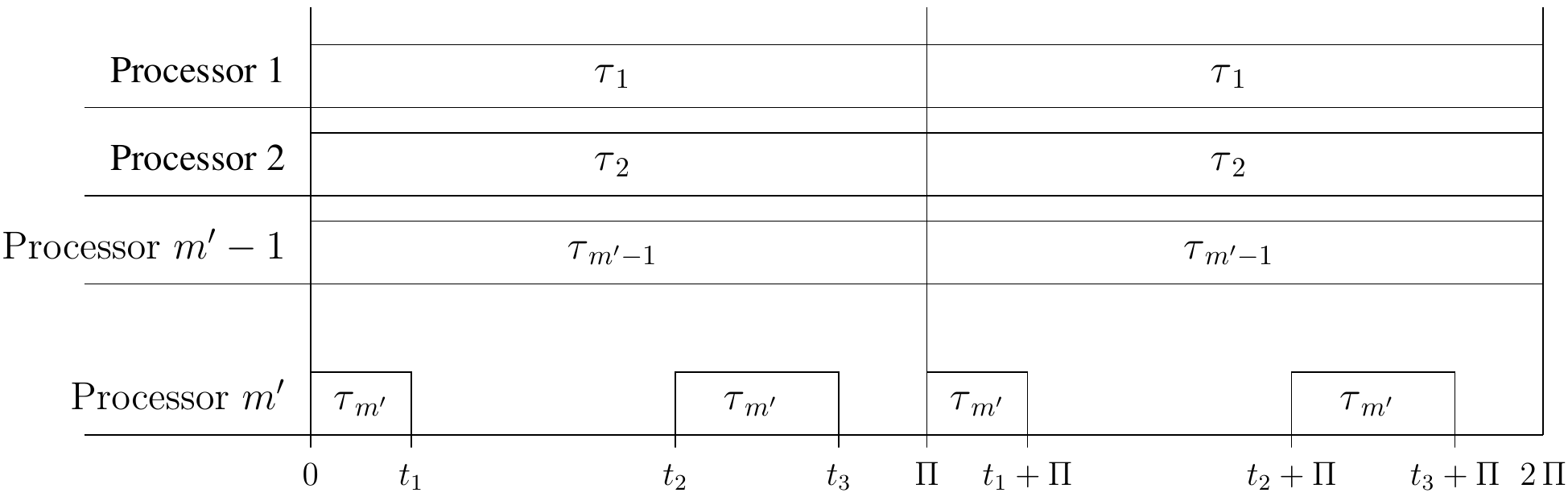}
\caption{McNaughton's schedule of implicit deadline periodic tasks
  with identical periods}
\label{fig:mpr:mcnaughton}
\end{figure}
Now suppose a cluster is abstracted into MPR interface $\GMpr$, which
is then transformed into task set $\taskset_{\mpr}$ using
Definition~\ref{def:mpr:demand-supply optimal_improved}. 
Let $\taskset_{\mpr}$ be scheduled on the multiprocessor platform
using McNaughton's algorithm, along with periodic tasks that all have
period and deadline $\Rperiod$ (implicit deadline task system with
identical periods). Figure~\ref{fig:mpr:mcnaughton} illustrates the 
McNaughton schedule for task set $\taskset_{\mpr}$. As can be seen in
the figure, tasks $\task_1, \ldots , \task_{\Rprocessors-1}$ completely utilize
$\Rprocessors-1$ processors on the platform. Further, every job of task  
$\task_{\Rprocessors}$ is scheduled in an identical manner within its execution
window (intervals $(0, t_1]$ and $(t_2, t_3]$ relative to release time).    
Since this schedule of $\taskset_{\mpr}$ is used as the processor supply for
the underlying MPR interface, $\mpr$ guarantees $\Rcapacity$ processor 
units in any time interval of length $\Rperiod$, $2 \Rcapacity$
processor units in any time interval of length $2 \Rperiod$, and so
on. In other words, the blackout interval of $\sbf_{\mpr}$ (described
in Section~\ref{sec:mpr:multi-core_resource_models}) reduces to
zero. The resulting $\sbf$ is plotted in
Figure~\ref{fig:mpr:sbf_graph_optimal} and it is given by the
following equation.
\begin{equation}
\sbf_{\mpr}(t) = \left \lfloor \frac{t}{\Rperiod} \right \rfloor
\Rcapacity + \left (t - \left \lfloor \frac{t}{\Rperiod} \right \rfloor
\Rperiod \right)\Rprocessors - \min \left \{ t - \left \lfloor \frac{t}{\Rperiod}
\right \rfloor \Rperiod, \Rprocessors \Rperiod - \Rcapacity \right \} 
\label{eqn:mpr:sbf_MPR_optimal}
\end{equation}
\begin{figure}
\centering
\includegraphics[width=0.8\linewidth]{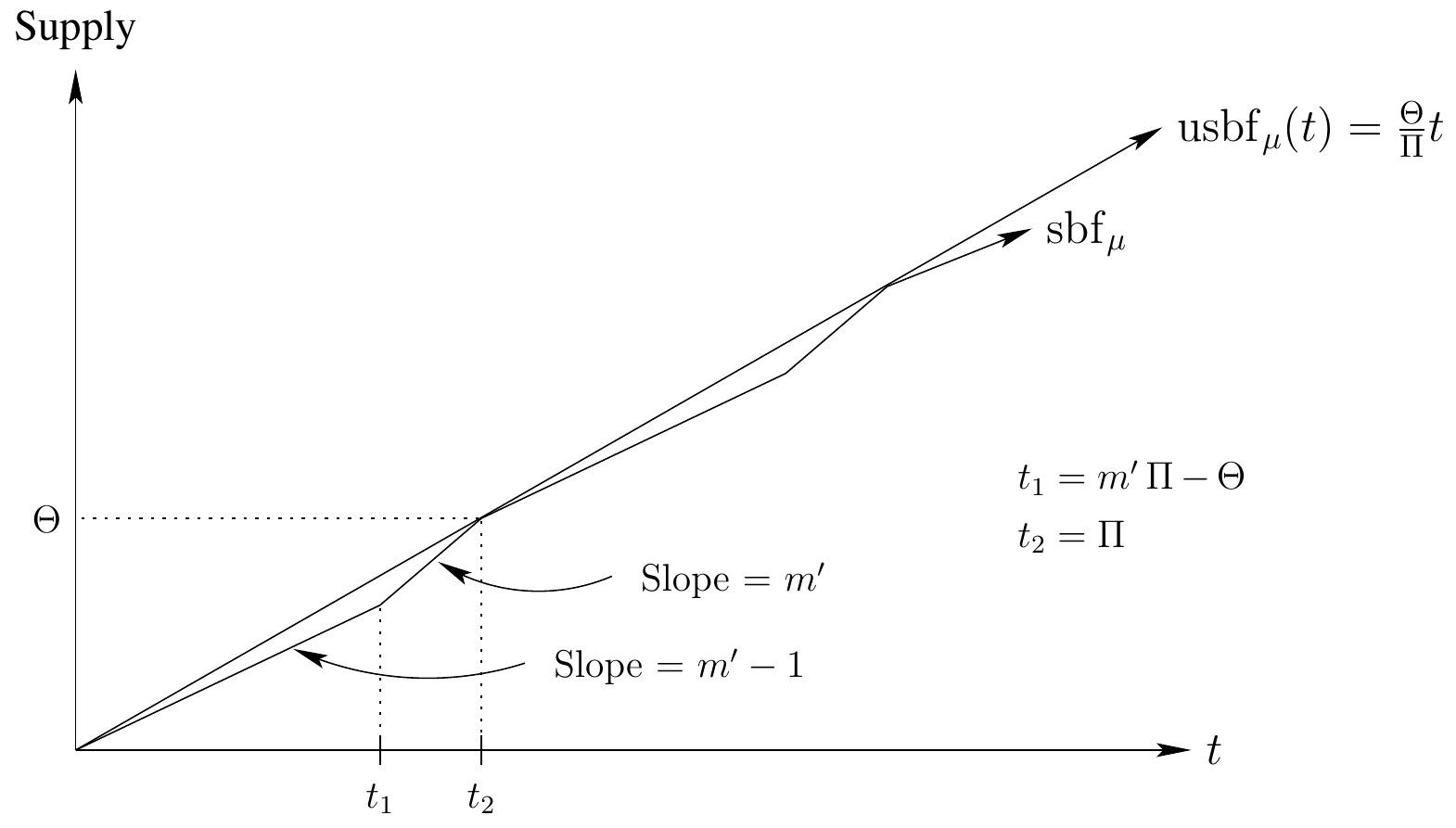}
\caption{Improved $\sbf_{\mpr}$ and its linear upper bound $\usbf_{\mpr}$}
\label{fig:mpr:sbf_graph_optimal}
\end{figure}
%\begin{figure}
%\centering
%\includegraphics[width=0.9\linewidth]{figures/mpr_sbf_schedule_MPR_optimal}
%\caption{Schedule of $\mpr$ w.r.t improved $\sbf_{\mpr}(t)$}
%\label{fig:mpr:sbf_optimal}
%\end{figure}

$\sbf_{\mpr}$ guarantees $\Rcapacity$ resource units in any time 
interval of length $\Rperiod$. Then a processor supply with supply
bound function equal to $\sbf_{\mpr}$ can successfully schedule task
set $\taskset_{\mpr}$. Thus we have eliminated the resource overhead 
that was present in the previous transformation given in
Definition~\ref{def:mpr:demand-supply optimal}.    

Now consider the schedulability condition for cluster $\component$ given by
Equation~\eqref{eqn:mpr:EDF schedulability} in
Theorem~\ref{thm:mpr:EDF schedulability}. This equation needs to be
evaluated for all values of $A_k$ up to the bound given in
Theorem~\ref{thm:mpr:Ak_bound} and for all tasks $\task_k$ in cluster
$\component$. In this equation it is easy to see that $\dem(A_k + 
\Tdeadline_k, \Rprocessors)$ increases by at most $\Rprocessors-1$ for
every unit increase in $A_k$, as long as $A_k + 1 + \Tdeadline_k$ does not
coincide with the release or deadline of some task in cluster
$\component$. In other words, $\dem(A_k + 1 + \Tdeadline_k,
\Rprocessors) \leq \dem(A_k + \Tdeadline_k, \Rprocessors) +
\Rprocessors-1$, whenever $A_k + 1 + \Tdeadline_k$ is not equal to $l
\Tperiod_i$ or $l \Tperiod_i + \Tdeadline_i$ for any $l$ and $i$
(denoted as property \emph{bounded increase}). This is because over
such unit increases in $A_k$, $\Rprocessors \Tcapacity_k$ and each 
$\hat{I}_{i,2}$ remain constant and
$\sum_{i: i \in L_{(\Rprocessors-1)}} (\bar{I}_{i,2} - 
\hat{I}_{i,2})$ increases by at most $\Rprocessors-1$. 
However $\sbf_{\mpr}$ increases by at least $\Rprocessors-1$ over
each unit time interval (see
Figure~\ref{fig:mpr:sbf_graph_optimal}). Therefore to generate
interface $\mpr$, it is sufficient to evaluate
Equation~\eqref{eqn:mpr:EDF schedulability} at only those 
values of $A_k$ for which $A_k + \Tdeadline_k$ is equal to $l
\Tperiod_i$ or $l \Tperiod_i + \Tdeadline_i$ for some $l$ and $i$.
Now suppose period $\Rperiod$ of $\mpr$ is equal to the
$\GCD$ (greatest common divisor) of the periods and deadlines of all
the tasks in cluster $\component$. Then all the required evaluations
of Equation~\eqref{eqn:mpr:EDF schedulability} will occur at time
instants $t$ for which $\sbf_{\mpr}(t) = \usbf_{\mpr}(t) =
\frac{\Rcapacity}{\Rperiod}t$ (see
Figure~\ref{fig:mpr:sbf_graph_optimal}). In other words, the
right hand side of Equation~\eqref{eqn:mpr:EDF schedulability} can be
replaced with $\frac{\Rcapacity}{\Rperiod}t$. This means that the
resource bandwidth of the resulting interface $\mpr$
($\frac{\Rcapacity}{\Rperiod}$) will be equal to the schedulability
load, $\max_k \max_{A_k} \frac{\dem(A_k+D_k,\Rprocessors)}{A_k+D_k}$,
of cluster $\component$. Thus we have eliminated the resource
overhead that was previously present in the cluster abstraction
process.

We now summarize the contributions of this section. The following
theorem, which is a direct consequence of the above discussions, states
the fundamental result of this section. This theorem states that our
improved virtual-clustering framework does not incur any resource
overheads in transforming MPR interfaces to periodic tasks or in
scheduling the transformed tasks on the multiprocessor platform.   
\begin{theorem}
Consider MPR interfaces $\mpr_1 = \tuple{\Rperiod, \Rcapacity_1,
  \Rprocessors_1} , \ldots , \mpr_p = \tuple{\Rperiod, \Rcapacity_p,
  \Rprocessors_p}$. Suppose they are transformed to
periodic tasks using Definition~\ref{def:mpr:demand-supply
  optimal_improved}. McNaughton's algorithm can successfully schedule
the transformed tasks on $m$ identical, unit-capacity processors
if and only if,
\begin{equation*}
\sum_{i = 1}^p \frac{\Rcapacity_i}{\Rperiod} \leq m 
\end{equation*}
\label{thm:mpr:optimal_inter_cluster_scheduling}
\end{theorem}

Suppose (1) we want to schedule a constrained deadline sporadic task
set $\taskset$ using virtual clusters on $m$ identical, unit-capacity
processors, (2) task-cluster mapping is given, and (3) each
intra-cluster scheduler is such that the corresponding schedulability
condition satisfies \emph{bounded increase} property described above
(\emph{e.g.}, \GEDF). Let (1) each virtual cluster be abstracted into
an MPR interface whose period $\Rperiod$ is equal to the $\GCD$ of the 
periods and deadlines of all the tasks in $\taskset$, (2) these interfaces
be transformed into periodic tasks using
Definition~\ref{def:mpr:demand-supply optimal_improved}, and (3) these
periodic tasks be scheduled on the multiprocessor platform using
McNaughton's algorithm. Then, in addition to the results stated in 
Theorem~\ref{thm:mpr:optimal_inter_cluster_scheduling}, the resource
bandwidth of each MPR interface will be equal to the schedulability
load of the corresponding cluster.
  
\subsection{Virtual clustering of implicit deadline task systems}
\label{sec:mpr:implicit_deadline_virtual_cluster}

In this section we propose two virtual cluster-based scheduling algorithms
for implicit deadline sporadic task sets. We consider the problem of scheduling 
an implicit deadline sporadic task set $\taskset = \{ \task_1 =
(\Tperiod_1, \Tcapacity_1, \Tperiod_1), \ldots , \task_n =
(\Tperiod_n, \Tcapacity_n, \Tperiod_n) \}$ on $m$ identical,
unit-capacity processors. We first present a new virtual-clustering technique
that is optimal like the well known Pfair algorithm~\citep{BCP96}, but
unlike Pfair, has a non-trivial bound on the number of
preemptions. The second technique extends the well known algorithm
\USEDF$\{m/(2m-1)\}$~\citep{SrBa02} with virtual clusters. We show that
the presently known processor utilization bound of \USEDF$\{m/(2m-1)\}$
can be improved by using virtual clusters.  

\subsubsection{VC-IDT scheduling algorithm} 
\label{sec:VC-IDT}

In VC-IDT (\emph{Virtual Clustering - Implicit Deadline Tasks}) scheduling
algorithm we consider a trivial task-processor mapping that assigns each
task $\task_i \in \taskset$ to its own virtual cluster $\component_i$
having one processor. Since each cluster has only one processor, we
assume that each cluster uses \EDF\ for intra-cluster
scheduling\footnote{Since each cluster also has only one task, any
  work conserving algorithm can be used for intra-cluster scheduling.}.
Each cluster $\component_i$ is abstracted into a MPR interface
$\mpr_i = \tuple{\Rperiod, \Rcapacity_i, 1}$, where $\Rperiod$ is
equal to the $\GCD$ of $\Tperiod_1, \ldots , \Tperiod_n$ and
$\Rcapacity_i/\Rperiod = \Tcapacity_i/\Tperiod_i$. Further, each
interface $\mpr_i$ is transformed into periodic tasks using
Definition~\ref{def:mpr:demand-supply optimal_improved} and the
resulting task set is scheduled on the 
multiprocessor platform using McNaughton's algorithm. The following
theorem proves that VC-IDT is an optimal algorithm for scheduling
implicit deadline sporadic task systems on identical, unit-capacity
multiprocessor platforms.
\begin{theorem}
Consider sporadic tasks $\taskset = \{
\task_1 = (\Tperiod_1, \Tcapacity_1, \Tperiod_1), \ldots , \task_n =
(\Tperiod_n, \Tcapacity_n, \Tperiod_n) \}$. A necessary and sufficient
condition to guarantee that $\taskset$ is schedulable on $m$ identical,
unit-capacity processors using VC-IDT algorithm is 
\begin{equation}
\sum_{i = 1}^n \frac{\Tcapacity_i}{\Tperiod_i} \leq m
\label{eqn:mpr:vc-idt_schedulability}
\end{equation} 
\label{thm:mpr:vc-idt_schedulability}
\end{theorem} 
\begin{proof}
In VC-IDT each virtual cluster $\component_i$, comprising of task
$\task_i$, is abstracted to interface $\mpr_i =
\tuple{\Rperiod, \Rcapacity_i, 1}$, where $\Rperiod$ is 
equal to the $\GCD$ of $\Tperiod_1, \ldots , \Tperiod_n$ and
$\frac{\Rcapacity_i}{\Rperiod} =
\frac{\Tcapacity_i}{\Tperiod_i}$. The interface set
$\mpr_1, \ldots , \mpr_n$, all having identical periods, are then
transformed to periodic tasks using
Definition~\ref{def:mpr:demand-supply optimal_improved} and scheduled
on the platform using McNaughton's algorithm. Therefore, from
Theorem~\ref{thm:mpr:optimal_inter_cluster_scheduling}, we get that 
this interface set is schedulable on the multiprocessor platform if
and only if,
\begin{align*}
& \sum_{i = 1}^n \frac{\Rcapacity_i}{\Rperiod} \leq m \\
\Rightarrow & \sum_{i = 1}^n \frac{\Tcapacity_i}{\Tperiod_i} \leq m
\end{align*}

To prove this theorem we then need to show that for each $i$,
interface $\mpr_i$ can schedule cluster $\component_i$. $\component_i$
comprises of sporadic task $\task_i$ and uses \EDF\
scheduler. Therefore any processor supply that can guarantee
$\Tcapacity_i$ processor units in all time intervals of length
$\Tperiod_i$ can be used to schedule $\task_i$. But from the $\sbf$
of model $\mpr_i$ (Equation~\eqref{eqn:mpr:sbf_MPR_optimal}), it is
easy to see that $\mpr_i$ guarantees $\Tcapacity_i$ processor units in
any time interval of length $\Tperiod_i$. This proves the theorem.   
\qed
\end{proof}

Equation~\eqref{eqn:mpr:vc-idt_schedulability} is known to be a
necessary and sufficient feasibility condition for scheduling implicit
deadline sporadic task systems on $m$ identical, unit-capacity
processors~\citep{SrAn02}. Hence VC-IDT is an optimal scheduling
algorithm for this problem domain. The other known optimal schedulers
for this problem, to the best of our knowledge, are the PD$^2$
Pfair/ERfair algorithm~\citep{SrAn02} and the task-splitting
algorithm~\citep{AnBl08}.

PD$^2$ algorithm is known to incur a high number of preemptions in order
to guarantee P-fairness/ER-fairness, because fairness is a stricter
requirement than deadline satisfaction. It can potentially incur $m$
preemptions in every time unit, which is the maximum possible on this
multiprocessor platform. In contrast, the number of preemptions
incurred by VC-IDT has a non-trivial upper bound which can be
explained as follows. When interfaces $\mpr_1, \ldots , \mpr_n$ are 
scheduled using McNaughton's algorithm (after being transformed into
periodic tasks), there are at most $m-1$ of them that use more than
one processor. Each such interface $\mpr_i$ is preempted once in every  
$\Rperiod$ time units and this may result in a preemption in the
execution of task $\task_i$. Each of the other $n-(m-1)$ tasks may
also experience preemption once in every $\Rperiod$ time units,
because the execution requirements of a job of this task cannot be
entirely satisfied by a single job of the corresponding
interface. The entire sporadic task set will thus incur at most $n$
preemptions in every $\Rperiod$ time units. Therefore when
$\Rperiod$, the GCD of task periods, is very small VC-IDT does not
offer any advantage over PD$^2$ algorithm. This can happen for
instance even if two task periods are co-prime (the GCD in this case 
is one). However, in real-world systems, it has been observed that
task periods are typically harmonic to (multiples of) each other. For
example, harmonic task periods can be found in avionics real-time
applications; see ARINC-653 standards~\citep{arinc653} and sample
avionics workloads in the appendix of this technical
report~\citep{techreport_rtcsa}. In this case, the GCD of task periods 
is equal to the smallest task period (typically a few milliseconds as
indicated by the workloads in~\citep{techreport_rtcsa}), and then
VC-IDT incurs far fewer preemptions than Pfair/ERfair algorithms. It
is worth noting that although the BoundaryFair algorithm~\citep{ZMM03}
incurs fewer preemptions than VC-IDT, it is only optimal for
scheduling periodic (not sporadic) task systems.

The task splitting algorithm proposed by Andersson and
Bletsas~\citep{AnBl08} has also been shown to be optimal for implicit 
deadline sporadic task systems (see Theorem~3
in~\citep{AnBl08}). Suppose $jobs(t)$ denotes the maximum number of
jobs that will be released by the task system in any time interval of
length $t$. Then this algorithm is known to incur at most
$\frac{3mt}{\GCD} + 2m + jobs(t)$ number of preemptions, where $\GCD$
denotes the greatest common divisor of task periods (derived from
Theorems~2 and~3 in~\citep{AnBl08}). In contrast, 
VC-IDT algorithm incurs at most $\frac{n t}{\GCD}$ number of
preemptions. Clearly, our algorithm outperforms the task splitting
approach whenever $n < 3m$. When $n > 3m$, either algorithm can
incur fewer preemptions depending on the value of $\GCD$ and the
relation between task periods. The runtime complexity of the
dispatcher under task splitting is the same as that of partitioned
$\EDF$ (roughly logarithmic in the number of tasks for every
scheduling decision). In contrast, under VC-IDT, the entire interface
schedule based on McNaughton's algorithm can be generated and stored
offline for intervals of length $\GCD$. Therefore at runtime the
tasks can be scheduled in constant time. This vastly improved runtime 
complexity at the expense of increased storage requirements is
particularly useful in embedded systems, where cheaper ROM and Flash
memory is still preferred over the more expensive RAM (for instance,
MICAz, the sensor node from crossbow, has 512k of Flash memory whereas
only 4k of RAM~\citep{micaz}). Finally, a practical limitation of the
task splitting approach is that they do not provide any error
isolation mechanism, \emph{i.e.}, a task that executes for more than
its stated worst-case execution time can cause other tasks in the
system to miss deadlines. In contrast, VC-IDT provides automatic
error isolation, because a mis-behaving task will never get more
processor share than already provided by its MPR interface.

\subsubsection{Virtual clustering for US--EDF$\{m/(2m-1)\}$}

\USEDF$\{m/(2m-1)\}$, proposed by Srinivasan and Baruah~\citep{SrBa02}, is
a global scheduling algorithm for implicit deadline sporadic task
systems. Under this algorithm each task with utilization
($\frac{\Tcapacity}{\Tperiod}$) greater than $\frac{m}{2m-1}$ is given
the highest priority, and the remaining tasks are scheduled based on
\GEDF. It has been shown that this algorithm has a processor
utilization bound of $\frac{m^2}{2m-1}$, \emph{i.e.}, any sporadic
task set with total utilization ($\sum_i
\frac{\Tcapacity_i}{\Tperiod_i}$) at most $\frac{m^2}{2m-1}$ can be
scheduled by \USEDF$\{m/(2m-1)\}$ on $m$ identical, unit-capacity
processors~\citep{SrBa02}.

Now consider the following virtual cluster-based \USEDF$\{m/(2m-1)\}$
scheduling algorithm. Let each task with utilization greater than
$\frac{m}{2m-1}$ be assigned to its own virtual cluster having one
processor and using \EDF\ (denoted as \emph{high utilization
  cluster}), and all the remaining tasks be assigned to a single
cluster using \GEDF\ (denoted as \emph{low utilization cluster}). Each
cluster is abstracted to a MPR interface such that period $\Rperiod$
of each interface is equal to the $\GCD$ of $\Tperiod_1, \ldots ,
\Tperiod_n$. Each high utilization 
cluster is abstracted to interface $\tuple{\Rperiod, \Rcapacity,
  1}$, where $\frac{\Rcapacity}{\Rperiod}$ is equal to the utilization of
task in the cluster (Theorem~\ref{thm:mpr:vc-idt_schedulability}
proves correctness of this abstraction). The low utilization cluster
is abstracted to interface $\mpr_{low} = \tuple{\Rperiod,
  \Rcapacity',\Rprocessors}$, where $\Rcapacity'$ and $\Rprocessors$ are
generated using techniques in 
Sections~\ref{sec:mpr:interface_generation} 
and~\ref{sec:mpr:improved_virtual_cluster}. Finally, these interfaces
are transformed to periodic tasks using
Definition~\ref{def:mpr:demand-supply optimal_improved} and the
resulting task set is scheduled on the multiprocessor platform using
McNaughton's algorithm.
 
We now derive a utilization bound for the virtual
cluster-based \USEDF$\{m/(2m-1)\}$ algorithm described above. 
Suppose $\alpha$ denotes the total utilization of all the high utilization
tasks, \emph{i.e.}, the total resource bandwidth of all the MPR interfaces that
represent high utilization clusters is $\alpha$. Since all the
interfaces that we generate have identical periods, from 
Theorem~\ref{thm:mpr:optimal_inter_cluster_scheduling} we get that the
maximum resource bandwidth available for $\mpr_{low}$ is
$m-\alpha$. This means that $\frac{\Rcapacity'}{\Rperiod} \leq
m-\alpha$ and $\alpha \leq m$ are necessary and sufficient conditions
to guarantee schedulability of task set $\taskset$ under virtual
cluster-based \USEDF$\{m/(2m-1)\}$.

Suppose $\alpha > m-1$. Then $m - \alpha < 1$ and $\Rprocessors \leq
1$. The last inequality can be explained as follows. $\Rprocessors =
\left \lceil \frac{\Rcapacity'}{\Rperiod} \right \rceil$ because
$\Rprocessors$ is the smallest number of processors upon which the low
utilization cluster 
is schedulable. Then $\frac{\Rcapacity'}{\Rperiod} \leq m-\alpha <
1$ implies $\Rprocessors \leq 1$. In this case the low utilization cluster is
scheduled on a uniprocessor platform and \GEDF\ reduces to \EDF, an
optimal uniprocessor scheduler with utilization bound $m -
\alpha$. Therefore virtual cluster-based \USEDF$\{m/(2m-1)\}$ is
optimal whenever $\alpha > m-1$, \emph{i.e.}, it can successfully
schedule task set $\taskset$ if $\sum_{i = 1}^n
\frac{\Tcapacity_i}{\Tperiod_i} \leq m$.

Now suppose $\alpha \leq m-1$. To derive the utilization
bound in this case, we use a utilization bound of \GEDF\ that was
developed by Goossens \emph{et al.}~\citep{GFB03}. As per this
bound $\mpr_{low}$ can support a low utilization cluster whose 
total task utilization is upper bounded by
$\left(\frac{\Rcapacity'}{\Rperiod}-\left(\frac{\Rcapacity'}{\Rperiod}-1 
  \right)U_{max}\right)$, where $U_{max}$ is the maximum utilization 
of any task in the cluster. Therefore, in this case, the
utilization bound of virtual cluster-based \USEDF$\{m/(2m-1)\}$ is 
\begin{align*}
& \alpha + \left(\frac{\Rcapacity'}{\Rperiod} -
\left(\frac{\Rcapacity'}{\Rperiod}-1 \right)U_{max}\right) \\
& = \alpha + \left(m - \alpha -
\left(m - \alpha - 1 \right)U_{max}\right)
\end{align*}
Since $m-\alpha \geq 1$, the bound in the above equation is 
minimized when $U_{max}$ is maximized. Substituting $U_{max} =
\frac{m}{2m-1}$ (largest utilization of any task in the low
utilization cluster), we get a utilization bound of 
\begin{align*}
& \alpha + \left(m-\alpha \left (1 - \frac{m}{2m-1} \right) +
\frac{m}{2m-1} \right) \\
& = \frac{\alpha (2m-1) + (m-\alpha )(m-1) + m}{2m-1} \\
& \geq \frac{\alpha (2m-1) + (m-\alpha )(m-1) + m}{2m-1} \\
& = \frac{m^2 + \alpha m}{2m-1}
\end{align*}
Thus the processor utilization bound of virtual cluster-based
\USEDF$\{m/(2m-1)\}$ is $\min\left \{ m, \frac{m^2 + \alpha m}{2m-1}
\right \}$. It 
is easy to see that whenever $\alpha > 0$, this bound is greater than 
the presently known utilization bound of $\frac{m^2}{2m-1}$ for
\USEDF$\{m/(2m-1)\}$. This shows that virtual clustering, unlike the
earlier \USEDF$\{m/(2m-1)\}$ algorithm, allows one to use the leftover
processing capacity from high utilization clusters for scheduling
tasks in the low utilization cluster. It also shows that the
improvement in utilization bound is achievable even when clusters are
scheduled on the platform using non-trivial abstractions such as MPR
models. This gain however comes at a cost; since $\Rperiod$ is equal
to the $\GCD$ of task periods, the resulting schedule can potentially
incur more preemptions when compared to the original algorithm.

\section{Conclusions}
\label{sec:mpr:conclusions}

In this paper we have considered the idea of cluster-based scheduling
on multiprocessor platforms as an alternative to existing partitioned
and global scheduling strategies. Cluster-based scheduling can be
viewed as a two-level scheduling strategy. Tasks in a cluster
are globally scheduled within the cluster (intra-cluster scheduling)
and clusters are then scheduled on the multiprocessor
platform (inter-cluster scheduling). We have
further classified clustering into physical (one-to-one) and virtual
(many-to-many), depending on the mapping between clusters and processors
on the platform. Virtual clustering is more general and less sensitive
to task-processor mappings than physical clustering.

Towards supporting virtual cluster-based scheduling, we have
developed techniques for hierarchical scheduling in this paper.
Resource requirements and concurrency constraints of tasks within
each cluster are first abstracted into MPR interfaces. These
interfaces are then transformed into periodic tasks which are used
for inter-cluster scheduling. We have also developed an efficient
technique to minimize processor utilization of individual clusters
under \GEDF. Finally, we developed a new optimal scheduling algorithm
for implicit deadline sporadic task systems, and also illustrated the
power of general task-processor mappings by virtualizing
\USEDF$\{m/(2m-1)\}$ algorithm.

We only focused on \GEDF\ for intra-cluster and McNaughton's
for inter-cluster scheduling. However, our approach of isolating the
inter-cluster scheduler from task-level concurrency constraints is
general, and can be adopted to other scheduling 
algorithms as well. Moreover, this generality also means that our
technique enables clusters with different intra-cluster schedulers to
be scheduled on the same platform. It would be interesting to 
generalize this framework by including other intra and inter-cluster
scheduling algorithms, with an aim to solve some open problems in
multiprocessor scheduling.

\begin{acknowledgements} 
The authors are grateful to the various anonymous reviewers of this
work. In particular, we would like to thank the reviewer who pointed
out the mistake in our $\sbf_{\mpr}$ formulation.
\end{acknowledgements}

%\bibliographystyle{plain}
%\bibliography{ecrts08_journal}

% BibTeX users please use one of
\bibliographystyle{plainnat}      % basic style, author-year citations
\bibliography{ecrts08_journal}   % name your BibTeX data base

\end{document}